\documentclass[twoside,11pt]{article}

\usepackage{amsgen,amsmath,amstext,amsbsy,amsopn,amsfonts,amssymb, amsthm}
\usepackage{color}
\usepackage{geometry}
\usepackage{enumitem}
\usepackage{graphicx,subfigure}
\usepackage[table]{xcolor}
\usepackage[ruled,linesnumbered]{algorithm2e}
\usepackage{booktabs}
\usepackage{hyperref}

\newtheorem{thm}{Theorem}
\newtheorem{lem}[thm]{Lemma}

\newtheorem{prp}[thm]{Proposition}
\newtheorem{df}[thm]{Definition}

\newtheorem{rem}[thm]{Remark}

\newcommand{\indic}{\mathrm{1}}
\def\n{\vert\vert}
\def\>{$\Rightarrow$}
\def\<>{$\Leftrightarrow$}
\def\G{\mathcal{G}}

\def\R{\mathbb{R}}
\def\L{\mathcal{L}}

\def\=>{\Rightarrow}
\def\spL{\mathrm{sp}(\L)}

\newcolumntype{L}{>{$}l<{$}}
\newcolumntype{C}{>{$}c<{$}}
\definecolor{lgray}{gray}{0.75}
\def\lg{\cellcolor{lgray}}
\definecolor{ggray}{gray}{0.85}
\def\gg{\cellcolor{ggray}}
\definecolor{agray}{gray}{0.95}
\def\ag{\cellcolor{agray}}

\providecommand{\keywords}[1]{{\textit{Keywords:}} #1}

\RequirePackage{authblk}
\usepackage{fancyhdr}
\title{Localized Fourier Analysis for Graph Signal Processing}
\author{
{Basile de Loynes\thanks{Basile de Loynes\\
\hspace*{1.8em}ENSAI, France, 
E-mail: basile.deloynes@ensai.fr
}
, Fabien Navarro\thanks{Fabien Navarro\\
\hspace*{1.8em}SAMM, Paris 1 Panth\'eon-Sorbonne University, France, 
E-mail: fabien.navarro@univ-paris1.fr
}
, Baptiste Olivier\thanks{Baptiste Olivier\\
\hspace*{1.8em}Orange Labs, France. 
E-mail: baptiste.olivier@orange.com}}
}

\begin{document}

\maketitle

\begin{abstract}
We propose a new point of view in the study of Fourier analysis on graphs, taking advantage of localization in the Fourier domain. For a signal $f$ on  vertices of a weighted graph $\G$ with Laplacian matrix $\L$, standard Fourier analysis of $f$ relies on the study of functions $g(\L)f$ for some filters $g$ on $I_\L$, the smallest interval containing the Laplacian spectrum ${\rm sp}(\L) \subset I_\L$. We show that for carefully chosen partitions $I_\L = \sqcup_{1\leq k\leq K} I_k$ ($I_k \subset I_\L$), there are many advantages in understanding the collection $(g(\L_{I_k})f)_{1\leq k\leq K}$ instead of $g(\L)f$ directly, where $\L_I$ is the projected matrix $P_I(\L)\L$. First, the partition provides a convenient modelling for the study of theoretical properties of Fourier analysis and allows for new results in graph signal analysis (\emph{e.g.} noise level estimation, Fourier support approximation). We extend the study of spectral graph wavelets to wavelets localized in the Fourier domain, called LocLets, and we show that well-known frames can be written in terms of LocLets. From a practical perspective, we highlight the interest of the proposed localized Fourier analysis through many experiments that show significant improvements in two different tasks on large graphs, noise level estimation and signal denoising. Moreover, efficient strategies permit to compute sequence $(g(\L_{I_k})f)_{1\leq k\leq K}$ with the same time complexity as for the computation of $g(\L)f$.
\end{abstract}

\keywords{
Nonparametric regression; Multiscale statistics; Variance estimation; Concentration inequalities; Graph signal processing; Spectral graph theory; Graph Laplacian; Harmonic analysis on graphs}

\section{Introduction}

Graphs provide a generic representation for modelling and processing data that reside on complex domains such as transportation or social networks. Numerous works combining both concepts from algebraic and spectral graphs with those from harmonic analysis (see for example \cite{chung1997spectral,coifman2006diffusion,belkin2008towards} and references therein) have allowed to generalize fundamental notions from signal processing to the context of graphs thus giving rise to Graph Signal Processing (GSP). For an introduction to this emerging field and a review of recent developments and results see \cite{shuman2013emerging} and \cite{ortega2018graph}. 
In general, two types of problems can be distinguished according to whether the underlying graph is known or unknown. The first case corresponds to the setup of a sampled signal at certain irregularly spaced points (intersections of a transportation network, nodes in a computer network, \ldots). In the second case, a graph is constructed from the data itself, it is generally interpreted as a noisy realization of one or several distributions supported by a submanifold of the Euclidean space. In this latter context, the theoretical submanifold is somehow approximated using standard methods such as $k$-NN, $\varepsilon$-graph and their Gaussian weighted versions. In any of these cases, the framework is actually similar: it consists of a graph (given by the application or by the data) and signals are real-valued functions defined on the vertices of the graph.   

Notions of graph Fourier analysis for signals on graphs were introduced and studied over the past several years~\cite{shuman2013emerging, sandryhaila2014discrete, shuman2016vertex, sardellitti2017graph}. The graph Fourier basis is given by the eigenbasis $(\chi_\ell)_\ell$ of the Laplacian matrix $\L$. The Graph Fourier Transform (GFT) consists in representing a signal $f$ in the Fourier basis $(\langle f, \chi_\ell\rangle)_\ell$, and by analogy with the standard case, the eigenvalues of $\L$ play the role of \textit{frequencies}. From this definition, it follows that many filtering techniques are written in terms of vectors $g(\L)f$, for some filter functions $g$ which act on the spectrum of $\L$ (scaling, selecting, ...). Fourier analysis on graphs has been successfully applied to many different fields such as stationary signals on graphs~\cite{perraudin2017stationary}, graph signal energy study~\cite{girault2018irregularity}, convolutional neural networks on graphs~\cite{defferrard2016convolutional}.

Graph wavelets are an important application of graph Fourier analysis, and several definitions of graph wavelets were proposed~\cite{crovella2003graph,coifman2006diffusion,gavish2010multiscale,leonardi2013tight,tanaka2014m,gobel2018construction}.
 When performing Fourier analysis of a signal, there is no guarantee that localization of a signal in the frequency domain (a.k.a Fourier domain) implies localization in the graph domain. This phenomenon is illustrated by the fact that the eigenvectors corresponding to the upper part of Laplacian spectrum tend to be more oscillating than those from the bottom of the spectrum (see for example \cite[Fig. 1.6, p. 28]{tremblay:tel-01078956} for an illustration). To overcome this problem, \cite{hammond2011wavelets} developed a fairly general construction of a frame enjoying the usual properties of standard wavelets: each vector of the frame is defined as a function $g(s\L)\delta_m$ (where $\delta_m$ is a Kronecker signal, having zero values at
every vertex except $m$) and is localized both in the graph domain and the spectral domain at fine scale $s$. The transform associated with this frame is named Spectral Graph Wavelet Transform (SGWT), and it was used in numerous subsequent works~\cite{susnjara2015accelerated, behjat2016signal, gobel2018construction}.

Signals which are sparse in the Fourier domain form an important class of graph signals. Indeed, there is a tight relationship between sparsity in the Fourier domain and the notion of regularity of a signal $f$ on the vertices of a graph $\G$ which comes from the Laplacian matrix $\L$ of $\G$. Intuitively, a smooth signal will not vary much between two vertices that are close in the graph. This regularity property can be read in the Fourier domain: a very smooth signal will be correctly represented in the Fourier domain with a small number of eigenvectors associated with the lower spectral values; on the contrary, non-smooth signals (\emph{i.e.} highly oscillating) are represented with eigenvectors corresponding to the upper part of the spectrum. Both the types of signal are said \emph{frequency sparse}.

In this paper, we propose to exploit localization in the Fourier domain to improve graph Fourier analysis. More precisely, we consider vectors of the form $g(\L_{I_k})f$ instead of vectors $g(\L)f$ in graph Fourier analysis, where $\L_{I_k}$ is defined as the matrix $\L P_{I_k}(\L)$ and $P_{I_k}(\L)$ denotes the projection onto the eigenspaces whose eigenvalue is contained in subset $I_k$. Localized Fourier analysis is motivated by problems and properties defined on strict subsets of the spectrum ${\rm sp}(\L)$ (\emph{e.g.} any problem defined in terms of frequency sparse graph signals). As a central application of Fourier localization, we introduce the Fourier localized counterpart of SGWT, that we call \textit{LocLets} for \textit{Loc}alized graph wave\textit{lets}. We prove that various frame constructions can be written in terms of LocLets, hence benefiting from all the advantages of localization discussed in this paper. 

Defining $I_\L$ as the smallest interval containing the entire spectrum ${\rm sp}(\L)$, the local Fourier analysis consists in choosing a suitable partition $I_\L = \sqcup_k I_k$ into subintervals on which standard Fourier analysis is performed. Such an analysis on disjoint intervals naturally benefits from several interesting properties. In particular, when $f$ is modeled by a Gaussian random vector with independent entries, the disjointness of subintervals preserves these properties in the sense that random variables $(g(\L_{I_k})f)_k$ are still Gaussian and independent. This simple observation has important consequences to study the graph problem at stake. In this work, it allows us to propose some noise level estimator from the random variables sequence $(g(\L_{I_k})f)_k$, and to provide a theoretical analysis of the denoising problem. Disjointness of subsets $(I_k)_k$ also provides simple strategies to parallelize Fourier analysis computations.

We also consider the general problem given by a noisy signal on a graph $\widetilde{f} = f + \xi$, where $\xi$ is some random Gaussian vector with noise level $\sigma$. We provide results for two important tasks: the estimation of $\sigma$ when the latter is unknown, and the denoising of noisy signal $\widetilde{f}$ in order to recover signal $f$. We show that for frequency sparse signals, localization allows to adapt to the unknown Fourier support of signal $f$. Theoretical guarantees and practical experiments show that localized Fourier analysis can improve state-of-the-art denoising techniques, not only in precision of the estimator $\widehat f$ of $f$, but also in time computations.

We provide an efficient method to choose a partition $I_\L = \sqcup_k I_k$ for the Fourier localized vectors $g(\L_{I_k})f$ to be sufficiently informative. Using well-known techniques for efficient graph Fourier analysis (a.k.a Chebyshev filter approximations), we propose scalable methods to perform localized Fourier analysis with no computational overhead over standard fast Fourier graph analysis. In particular, all methods introduced in the paper avoid the computation of the entire eigendecomposition of $\L$ which is a major computational advantage when considering large graphs in applications.

The paper is structured as follows. Section~\ref{sec-locfour} presents the relevant notions and techniques necessary to perform localized Fourier analysis. In Section~\ref{sec-loclets}, we introduce LocLets, the Fourier localized extension of SGWT. Section~\ref{sec-denoising} is devoted to the study of the denoising problem for signals on graphs. The section provides results about noise level estimation, and signal denoising. Additional properties of LocLets, such as computational aspects and relationships with known wavelet transforms, are further developed in Section~\ref{sec:chebyshev}. In Section~\ref{experiments}, we analyze the experiments made to support the interesting properties of localized Fourier analysis highlighted in this paper. Finally, the proofs are gathered in Section~\ref{sec:proofs}.


\section{Localized Fourier analysis for graph signals}
\label{sec-locfour}

In this section, we introduce the central notion studied in this paper: localization of graph Fourier analysis. First, we recall the relevant notions of graph Fourier analysis. Then we provide examples from previous works that motivates localization in the Fourier domain. We also introduce LocLets, an important application of Fourier localization to SGWT. Finally, we discuss briefly the particular case of graphs sampled from a manifold. 

\subsection{Functional calculus and Fourier analysis for graph signals}

Let $\G=(\mathcal V,\mathcal E)$ be an undirected weighted graph with $\mathcal V$ the set of vertices, $\mathcal E$ the set of edges, $n=|\mathcal V|$ the number of nodes (the \textit{size} of $\G$), and $(W_{ij})_{i,j\leq n}$ the weights on edges. Let us introduce the diagonal degree matrix whose diagonal coefficients are given by $D_{ii}=\sum_{1 \leq j \leq n} W_{ij}$ for $1 \leq i \leq n$. The resulting non-normalized Laplacian matrix $\L$ of graph $\G$ is defined as $\L = D-W$. The $n$ non-negative eigenvalues of $\L$, counted without multiplicity, are denoted by $\lambda_1, \ldots, \lambda_n$ in the decreasing order. In the sequel, $\spL$ stands for the spectrum of $\L$. The corresponding eigenvectors are denoted $\chi_1, \ldots, \chi_n$.

Given a graph $\G$, the GFT of a real-valued function $f$ defined on the vertices of $\G$ is nothing but the representation of $f$ in the orthonormal basis of eigenvectors of $\L$. Namely, for a signal $f:\G\rightarrow \R$, the $\ell$-th Fourier coefficient of $f$, denoted $\widehat f(\ell)$, is given by $\widehat f (\ell) = \langle f, \chi_{\ell} \rangle$. The Fourier support ${\rm supp}(\widehat f)$ of signal $f$ is the set of indices $\ell$ such that $\widehat f(\ell) \neq 0$. We will see in Section \ref{sec-loclets} that graph wavelets can be defined in a similar manner.

Functional calculus is a powerful technique to study matrices, and constitutes the heart of GSP. For a function $g$ defined on some domain $D_g$, $\spL \subset D_g$, functional calculus reads as
\begin{equation*}
g(\L) = \sum_{1 \leq \ell \leq n} g(\lambda_{\ell}) \langle \chi_{\ell}, \cdot \rangle \chi_{\ell}.
\end{equation*}
Interpreting the eigenvalues $\lambda_\ell$, $\ell=1, \ldots, n$, as the fundamental frequencies associated with a graph, the linear map $g(\L)$ is generally seen as a filter operator in terms of signal analysis. 

Also, spectral projections of matrix $\L$ can be made explicit with the help of functional calculus, by setting $g=\indic_I$. More precisely, for any subset $I\subset I_\L$, consider the map $P_{I}(\L)$ given by: 
\begin{equation*}
P_{I}(\L) = \sum_{1 \leq \ell \leq n} \indic_I(\lambda_\ell) \langle\chi_{\ell}, \cdot \rangle \chi_{l} = \sum_{\ell:~\lambda_{\ell} \in I} \langle\chi_{\ell}, \cdot \rangle \chi_{\ell}.
\end{equation*}
Then, $P_I(\L)$ is nothing but the spectral projection on the linear subspace spanned by the eigevectors associated with the eigenvalues belonging to $I$. In the sequel, $n_I = | I \cap \spL |$ will stand for the number of eigenvalues contained in subset $I \cap \spL$.

Spectral projections are a practical tool to focus on some part of the spectrum $\spL$. More precisely, let $I_\L = \sqcup_{1\leq k \leq K} I_k$ be a partition of interval $I_\L$ into disjoint subsets $(I_k)_k$. Since intervals $I_k$ are disjoints, functional analysis of $\L$ reduces to that of its projections $\L_{I_k} = \L P_{I_k}(\L)$ in the sense of the identity: 
\begin{equation*}
g(\L)= \sum_{1 \leq k \leq K} g(\L_{I_k}).
\end{equation*}

In this paper, one will study the extent to which Fourier analysis on large graphs is improved when considering \emph{local} Fourier analysis on each subset $I_k$ instead of \emph{global} Fourier analysis on $I_\L$.

\subsection{LocLets: a localized version of SGWT}

\label{sec-loclets}

This section introduces an important application of the localized graph Fourier analysis, namely the notion of localized SGWT.

\subsubsection{Construction of a SGWT}

Let $f:\G\rightarrow \R$ be a signal on the graph $\G$. Let $\varphi, \psi:\R \rightarrow\R$ be respectively the scaling and kernel functions (a.k.a. \textit{father} and \textit{mother} wavelet functions), and let $s_j > 0$, $1\leq j\leq J$, be some scale values. The discrete SGWT is defined in \cite{hammond2011wavelets} as follows:
\begin{equation*}
\mathcal W f = ( \varphi(\L)f^{T} , \psi(s_{1}\L)f^{T}, \ldots, \psi(s_J\L)f^T)^{T}.
\end{equation*}

The adjoint matrix $\mathcal W^{*}$ of $\mathcal W$ is:
\begin{equation}
\label{equ-adjointSGWT}
\mathcal W^{*}(\eta_{0}^{T}, \eta_{1}^{T}, \ldots, \eta_J^T)^{T} = \varphi(\L)\eta_{0} + \sum_{j=1}^J \psi(s_{j}\L)\eta_{j}.
\end{equation}
We also recall from \cite{hammond2011wavelets} that a discrete transform reconstruction formula using SGWT coefficients $(c_{j,m})_{\substack{0\leq j\leq J\\ 1\leq m\leq n}}$ is obtained by the formula
\begin{equation*}
(\mathcal W^{*} \mathcal W)^{-1} \mathcal W^{*}(c_{j,m})_{j,m},
\end{equation*}
where $(\mathcal W^{*} \mathcal W)^{-1}$ stands for a pseudo-inverse of the matrix $\mathcal W^{*} \mathcal W$.

\begin{rem}
From a theoretical point of view, no more assumptions on the scale values $s_j$ are required. However, the choice in practice of the scale values depends simultaneously on the mother wavelet $\psi$ and the maximal spectral value of the Laplacian $\L$. In \cite{hammond2011wavelets}, these values $s_j$ are suitably chosen accordingly with the graphs considered in the experiments. Besides, for our experiments, we follow the construction of \cite{gobel2018construction} in which $s_j=b^j$ for some real $b>1$ (see also Section \ref{sec-ParsevalFrame} for the details).  
\end{rem}

\subsubsection{Motivation for considering Fourier localized graph signals}

The definition of SGWT as given in~\cite{hammond2011wavelets} is closely related to its counterpart from traditional wavelet transform by the action of the transform on the Fourier domain. Equation (9) in~\cite{hammond2011wavelets} highlights the following decomposition of wavelet coefficients in the frequency domain:
\begin{equation}
\label{equ-coeffInFourier1}
\psi(s\mathcal L)f(x)
=
\frac{1}{2\pi} \int_{-\infty}^{+\infty}
  \hat \psi^* (s \omega) \hat f(\omega) e^{i \omega x} d\omega ,
\end{equation}
where $ \hat \psi^* (s \omega)=\hat{ \overline{\psi}_s }(\omega)$ denotes the Fourier coefficient of a scaled version $\overline{\psi}_s$ of the mother wavelet $\psi$, and $x\mapsto e^{i \omega x}$ are the eigenfunctions of the one-dimensional Laplacian $\mathcal L = \frac{\partial^2}{\partial x^2}$ on the real line. On the other hand, a similar formula holds for the case of graph wavelets where integration over frequencies is replaced by summation over Laplacian eigenvalues. It is given by Equation (22) in~\cite{hammond2011wavelets}:
\begin{equation}
\label{equ-coeffInFourier2}
\psi(s\mathcal L)f(n)
=
\sum_\ell \psi(s \lambda_\ell) \hat f(\ell) \chi_\ell(n).
\end{equation} 
Equations~\ref{equ-coeffInFourier1} and~\ref{equ-coeffInFourier2} suggest strong similarities between the role of frequencies for wavelets on the real line and the role of Laplacian eigenvalues for graph wavelets. For example, an explicit mapping between frequencies and Laplacian eigenvalues is exhibited in part III.C of~\cite{leonardi2013tight} for the case of the cycle graph. As an important consequence, filtering in the frequency domain for traditional time signals is similar to filtering Laplacian spectrum in the graph setting. This similarity has guided authors of~\cite{leonardi2011wavelet} who propose an analog of Meyer wavelets~\cite{meyer1985principe} for graphs by applying traditional Meyer filters on the graph Fourier domain. As for the traditional case, smoothness for graph Meyer wavelets can be guaranteed by calibrating Meyer filters so that the Fourier support of wavelet functions is concentrated in the bottom of the Laplacian spectrum. A larger list of tight frames candidates obtained from traditional wavelets defined on the Fourier domain is provided in Table 1 of~\cite{leonardi2013tight}. 

Other examples of signals with limited frequency support have inspired graph counterparts with restricted Fourier supports. In~\cite{anis2014towards, chen2015discrete, puy2018random, ricaud2019fourier}, $k$-bandlimited graph signals are supported on the eigenspaces associated with the $k$ smallest eigenvalues and define an equivalent of $\omega$-bandlimited signals. Compressive sampling is an application of this type of signals both in the traditional~\cite{candes2006stable, donoho2006compressed, candes2007sparsity} and in the graph~\cite{puy2018random} settings. While low-pass frequency filters are usually used to capture the smooth part of a signal, high-pass filters have the ability to point out anomalies in smooth signals. Such applications of high-pass filters can be found in~\cite{chen2008multi} for optical laser measurements and time signals, but also in~\cite{sandryhaila2014discrete} for anomalies detection in temperature measurements on a graph of weather stations. The latter application was also considered in Section V.D of~\cite{chen2015discrete} to illustrate the use of filter banks that split the graph signal into two bandlimited signals. Authors of~\cite{tanaka2018spectral, tanaka2020generalized} introduce graph sampling methods in the Fourier domain taking advantage of periodic patterns and extending to the graph setting sampling techniques designed for shift-invariant signals~\cite{eldar2009beyond}. The sampling task for time signals was addressed in \cite{herley1999minimum, lu2008theory} in the case of multiband signals supported on a disjoint union of frequency subsets. This suggests providing methods adapted to disjoint subsets in the Fourier domain also for the case of multiband graph signals. 

All the examples discussed above share a common feature: their Fourier support is localized in the Fourier domain in the sense that such signals $f$ satisfy $f=P_I(\mathcal L)f$ for some subset $I\subset I_{\mathcal L}$, $I\cap \rm{sp}(\mathcal L)\neq \rm{sp}(\mathcal L)$. For $k$-bandlimited signals, $I$ is a set containing only the $k$ smallest eigenvalues of $\mathcal L$; for high-pass filters, $I$ contains only the largest eigenvalues; for periodic signals, $I$ is a set such that $I\cap \rm{sp}(\mathcal L) = \{ \lambda_{i~\rm{ mod} ~m} ~|~ \lambda_i \in \rm{sp}(\mathcal L) \}$ for some period $m$; and for multiband signals, $I$ is a disjoint union $I=\sqcup_k I_k$ of intervals $(I_k)_k$. The current paper proposes methods to study these graph signals that we call Fourier localized signals.

\subsubsection{Definition of LocLets}

Spectral graph wavelet functions are given by $(\varphi(\L)\delta_m , \psi(s_j \L)\delta_m )_{1\leq j\leq J, 1\leq m\leq n}$. We define a LocLet function to be the projection of a graph wavelet function onto a subset of eigenspaces of $\L$. 

\begin{df}
  Let $(\varphi(\L)\delta_m , \psi(s_j \L)\delta_m )_{1\leq j\leq J, 1\leq m\leq n}$ be the family functions induced by a SGWT. Then, for any subset $I \subset I_\L$, $1\leq j\leq J$ and $1\leq m\leq n$, set:
\begin{displaymath}
\begin{split}
 & \varphi_{m,I} = \varphi(\L_{I})\delta_{m} = \sum_{\ell: \lambda_{\ell}\in I}  \varphi(\lambda_{\ell}) \widehat{ \delta}_m (\ell) \chi_{\ell}
 \\
  &\psi_{j,m,I} = \psi(s_{j} \L_{I})\delta_{m} = \sum_{\ell: \lambda_{\ell}\in I}  \psi(s_{j}\lambda_{\ell}) \widehat{ \delta}_m (\ell) \chi_{\ell},
\end{split}
\end{displaymath}
where $\widehat{\delta_m}(\cdot) = \langle \delta_m, \chi_\cdot \rangle$ is the graph Fourier transform of $\delta_m$. The functions $(\varphi_{m,I}, \psi_{j,m,I})_{1\leq j\leq J, 1\leq m\leq n}$ are called Localized waveLets functions (LocLets). The functions $\varphi_{m,I}, \psi_{j,m,I}$ are said to be localized at $I$.
\end{df}

Let $I_\L = \sqcup_{1\leq k\leq K} I_k$ be some partition. Then, the localized SGWT transfom of $f$ with respect to partition $(I_k)_{1\leq k\leq K}$, denoted by $\mathcal W^{(I_k)_k}f$, is defined as the family $\mathcal W^{(I_k)_k}f = (\mathcal W^{I_k}f)_k$ where
\begin{equation*}
  \mathcal W^{I_k}f = ( \varphi(\L_{I_k})f^{T} , \psi(s_{1}\L_{I_k})f^{T}, ...  )^{T}, \quad 1 \leq k \leq K.
\end{equation*}
Similarly to Equation~\eqref{equ-adjointSGWT}, the adjoint transform is given by
\begin{equation*}
  \mathcal W^{I_k\ast}(\eta_{0}^{T}, \eta_{1}^{T}, \ldots, \eta_J^T)^{T} = \varphi(\L_{I_k})\eta_{0} + \sum_{j=1}^J \psi(s_{j}\L_{I_k})\eta_{j}, \quad 1 \leq k \leq K.
  \end{equation*}
As already observed, localized SGWT of a signal $f$ contains more precise information about signal $f$ than its standard SGWT. The latter can easily be obtained from the former since subsets $(I_k)_k$ are pairwise disjoint and formula $g(s\L) = \sum_{1\leq k\leq K}g(s\L_{I_k})$ holds for all filter $g$, and in particular for $g=\varphi$ or $g=\psi$. When the partition $I_\L = \sqcup_{1\leq k\leq K} I_k$ is carefully chosen, we show that the SGWT localization provides interesting features such as independence of random variables in denoising modelling, or considerable improvements in denoising tasks.

\begin{rem}
The functions $\varphi_{m,I}, \psi_{j,m,I}$ are localized in the Fourier domain in the sense that the support of their Fourier transforms are contained in a subset $I$ of $I_{\L}$. A different localization property, observable in the graph domain, is considered in~\cite{hammond2011wavelets}. We refer to the latter property as \emph{graph domain localization} in the current paper. 
  
A property about the graph domain localization at fine scales is stated in \cite[Theorem 5.5]{hammond2011wavelets} but this result appears to be not informative in general for the case of Fourier localized functions $\varphi_{m,I}, \psi_{j,m,I}$. For instance, any function of the form $\psi 1_I$ vanishes in a neighborhood of $0$ as soon as $0\notin I$, as observed in Section 2.3 of~\cite{gobel2018construction}. However, other graph domain localization results were obtained by the authors of~\cite{coulhon2012heat} for the case of frames considered in~\cite{gobel2018construction} and discussed in our Section~\ref{sec-ParsevalFrame}. In addition, the weaker graph localization property~\cite[Lemma 5.2]{hammond2011wavelets} for powers of the Laplacian still holds in our Fourier localized setting since in practice we approximate any  function $x \to \psi(s_j x) 1_I(x)$ by a Chebyshev polynomial of order $N$ (see Section \ref{sec:chebyshev}).
\end{rem}

\subsection{Fourier localization for weighted graphs sampled from manifolds}

In some applications, the underlying graph is unknown and is built from the data. In this case, the resulting graph is thought as a random sampling of a low-dimensional sub-manifold embedded in a higher dimensional Euclidean space. 

More precisely, let $\mathcal M$ be a Riemannian manifold of dimension $d$ embedded in $\mathbb R^m$ with $m > d$. A popular way to define a graph from a finite set of points $\{ x_1, \ldots, x_n \} \subset \mathcal M$ consists in defining a weighted adjacency matrix $W=(W_{ij})_{i,j \leq n}$ as follows:
\begin{equation} \label{eq:weight-gaussian}
W_{ij}=k \left ( \frac{\|x_i-x_j\|^2_2}{2 \varepsilon} \right ) ,
\end{equation}
where $\| \cdot \|_2^2$ stands for the Euclidean distance in $\mathbb R^m$ and $\varepsilon > 0$ is some parameter called the bandwidth of the kernel $k$. A typical choice for the kernel $k$ is the exponential function $k(x)=\exp(-x)$, $x \in \mathbb R$. As an example, the \emph{swissroll} graph of Section \ref{experiments} is built following this idea.

A whole part of the literature is dedicated to the question of the convergence of the discrete (normalized or non-normalized) Laplacian matrices $L_{n,\varepsilon}$ toward the Laplace-Beltrami operator $\Delta_{\mathcal M}$ (see \cite{Ros:97} for a detailed exposition of this classical object from differentiable geometry). The discretized operator $L_{n,\varepsilon}$ depending on two parameters, the convergence as $n \to \infty$ and/or $\varepsilon \to 0$ have been considered in \cite{HeiAudLux:05,Sin:06,CoiLaf:06,GinKol:06,belkin2008towards,MarCoi:19}. Loosely speaking, theses results are devoted to the approximation of $\Delta_{\mathcal M}f$ by $L_{n,\varepsilon} f$ at the sample points.

Furthermore, it is shown in \cite{LuxBelBou:08,GarGerHeiSle:20} that, under mild conditions, eigenvalues and eigenfunctions of $\Delta_{\mathcal M}$ are well approximated by those of $L_{n,\varepsilon}$. As a consequence of particular interest, bandlimited (or even multiband) signals on the manifold $\mathcal M$, when sampled at the points $\{ x_1, \ldots, x_n \} \subset \mathcal M$, are expected to be bandlimited (multiband) with respect to the graph Laplacian $L_{n,\varepsilon}$ (with slight differences when considering the normalized or the non-normalized Laplacian). It is worth noting that the sample points $\{ x_1, \ldots, x_n \}$ do not have to lie exactly in the manifold $\mathcal M$ but can be disrupted by a noise. The spectral properties are preserved by standard spectral perturbation arguments (see \cite{Kat:95}) providing the noise level is sufficiently small. Such a perturbation argument is discussed at some point in \cite{CoiLaf:06} and remains valid in our context. 

To conclude this discussion, let us point out that the choice of a Gaussian kernel in \eqref{eq:weight-gaussian}, while popular, is quite arbitrary. The results in \cite{GarGerHeiSle:20} are stated for a rather general kernel including non-smooth kernels. In addition, a variable bandwidth kernel is also considered in \cite{BerHar:16}.


\section{Local Fourier analysis and graph functions denoising}
\label{sec-denoising}

The denoising problem is stated as follows: given an observed noisy signal $\widetilde{f}$ of the form $\widetilde{f}=f+\xi$ where $\xi$ is a $n$-dimensional Gaussian vector distributed as $\mathcal N(0,\sigma^2 \mathrm{Id})$, provide an estimator of the \emph{a priori} unknown signal $f$. 

This section shows how localized Fourier analysis helps in estimating the noise level $\sigma$ when it is unknown, and in recovering the original signal $f$ when the latter is sparse in the Fourier domain. In what follows, we will focus on random variables of the form $\|P_{I_k}\widetilde f_{I_k}\|_2$ where $\widetilde f$ is the noisy signal and $I_k$ is a subset in the partition $I_\L = \sqcup_k I_k$. To keep the notations light, $n_k$, $f_k$, $\xi_k$ and $\widetilde f_k$ will stand for $n_{I_k}$, $P_{I_k}f$, $P_{I_k}\xi$ and $P_{I_k}\widetilde f$ respectively. In addition, the cumulative distribution function of a random variable $X$ will be denoted by $\Phi_X$.


\subsection{Noise level estimation for frequency sparse signals}

Since in real application the noise level $\sigma$ remains unknown in general, new estimators $\widehat \sigma$ based on localization properties in the spectrum are introduced in the sequel.    

\subsubsection{Noise level estimation from projections along \texorpdfstring{${\rm sp}(\L)$}{spL}}

For any filter $g$ defined on $I_\L$ and any subset $I\subset I_\L$, simple computations give rise to
\begin{equation}
\label{eq-dirichletLike}
\mathbb E( \widetilde f^T g(\L_I) \widetilde f )=f^Tg(\L_I)f+\sigma^2 {\rm Tr}(g(\L_I)).
\end{equation}
Since both $\widetilde f^T g(\L_I) \widetilde f$ and ${\rm Tr}(g(\L_I))$ are known, Equation \eqref{eq-dirichletLike} suggests building estimators from the expression $\frac{\widetilde f^T g(\L_I) \widetilde f}{ {\rm Tr}(g(\L_I)) }$. In~\cite{de2019data}, the noise level is estimated by $\frac{\widetilde f^T \L \widetilde f}{ {\rm Tr}(\L) }$ which can be seen as the graph analog of the Von Neumann estimator from~\cite{vNe:41}. The main drawback of this estimator is its bias.

Theoretically, without any assumption on the signal $f$, the bias term $\frac{f^T g(\L_I) f }{{\rm Tr}(g(\L_I))}$ is minimized when $g=\indic_{\{\lambda_{\ell^\ast} \}}$ where $\ell^\ast = \mathrm{argmin} \{ | \widehat f(\ell) | : \lambda_\ell \in \spL \}$. The computation of such filters would require the complete reduction of $\L$ which does not scale well with the size of the graph. Instead, these ideal filters will be approximated by filters of the form $g=\indic_{I_k}$, for $I_k$ a subset in the partition $I_\L = \sqcup_k I_k$. It is worth noting that with $k^\ast = {\rm argmin}_{k} \|f_{k}\|_2$, the function $g^\ast=\indic_{I_{k^\ast}}$ achieves the minimal bias of the estimator among all filters of the form $g = \sum_k \alpha_k \indic_{I_k}$. 

Discarding some intervals $I_k$ with $n_k=0$, it can be assumed without loss of generality that $n_k\neq 0$ for all $1\leq k\leq K$. Also, observe that the random variable $\|\widetilde f_{k}\|_2^2$ can be decomposed as follows
\begin{equation}
\label{equ-noisyNorm}
\|\widetilde f_{k}\|_2^2=\| f_{k} \|_2^2  +  \| \xi_{k} \|_2^2+2 \langle f_{k} , \xi_{k} \rangle,
\end{equation}
where $\frac{\| \xi_{k} \|_2^2}{\sigma^2}$ and $\frac{\langle f_{k} , \xi_{k} \rangle}{\sigma}$ are random variables distributed as $\chi^2( n_k )$ and $\mathcal N( 0 , \| f_{k} \|_2^2 )$ respectively.

\begin{prp} \label{prp:block-distrib}
  Let $(c_k)_{1 \leq k \leq K}$ be the sequence of non-negative random variables defined, for all $k=1, \ldots, K$, by $c_k=\|\widetilde f_{k}\|_2^2/n_k$. Then,
\begin{enumerate}
\item the random variables $c_1, \ldots, c_K$ are independent;
\item for all $k,k^\prime$ such that $f_k=f_{k^\prime}$, $c_k$ and $c_{k^\prime}$ are identically distributed if and only if $n_k = n_{k'}$;
\item for $k$ such that $f_k=0$, $c_k$ is distributed as $\frac{\sigma^2}{n_k} \Gamma_{n_k}$ where $\Gamma_{n_k} \sim \chi^2(n_k)$.
\end{enumerate}
\end{prp}

\subsubsection{The case of frequency-sparse signals}

When the signal $f$ is sparse in the Fourier domain, the condition $f_k=0$ is met for most of the intervals $I_k\subset I_\L$. Let us define $I_f = \sqcup_{k: I_k \cap {\rm supp}\widehat f \neq \emptyset} I_k$ to be the union of subsets $I_k$ intersecting the Fourier support ${\rm supp} (\widehat f)$ of $f$. Also, denote by $\overline{I_f} = I_L \backslash I_f$ its complement set. In order to take advantage of Fourier sparsity, let us introduce the quantities $\widehat \sigma_{\rm mean}$ and $\widehat \sigma_{\rm med}$ as follows: 

\begin{equation}
\label{equ-simpleEstimators}
\widehat \sigma_{\rm mean}(c)^2 = \frac{1}{|\{k: I_k \subset \overline{I_f} \}|} 
\sum_{k: I_k \subset \overline{I_f}} c_k \quad \textrm{and} \quad \widehat \sigma_{\rm med}(c)^2 = \mathrm{median}_{k: I_k \subset \overline{I_f}}(c_k).
\end{equation}

The following concentration inequalities show that $\widehat \sigma_{\rm mean}$ and $\widehat \sigma_{\rm med}$ are natural estimators of the noise level $\sigma$. 

\begin{prp} \label{concentration-noiseLevel}
Let $K_f = |\{k: I_k \subset \overline{I_f} \}|$, $n_0=\min \{ n_k: k, I_k \subset \overline{I_f} \}$, $n_\infty=\max \{ n_k:k, I_k \subset \overline{I_f} \}$, $V_f = 2\sigma^4 \sum_{k:I_k\subset \overline{I_f}} {1}/{n_k}$ and $B_f = {2\sigma^2}/n_0$. Then the following concentration inequalities hold:
\begin{enumerate}
\item for all $t\geq 0$,
\begin{displaymath}
\mathbb P\left(
 \widehat \sigma_{\rm mean}(c)^2 - \sigma^2 \geq t 
 \right)
\leq
\exp\left(
- \frac {K_f^2 t^2} {V_f ( 1 + B_f + \sqrt{ 1 + \frac{2 B_f K_f t}{V_f} } )}
\right),
\end{displaymath}
and for all $0\leq t\leq \sigma^2$,
\begin{displaymath}
\mathbb P\left(
 \widehat \sigma_{\rm mean}(c)^2 - \sigma^2 \leq - t 
 \right)
\leq
\exp\left(
- \frac {K_f^2 t^2}{2 V_f}
\right);
\end{displaymath}
\item for all $t\geq 0$, with $\beta=n_0/n_\infty$,
  \begin{displaymath}
    \mathbb P \left ( \widehat \sigma^2_{\rm med} 
\geq    
     \beta^{-1} \sigma^2 + 2\sigma^2 \beta^{-1} t \right ) \leq \exp \left( \frac{K_f}{2} \ln \Big [ 4p^+(t)(1-p^+(t)) \Big ] \right),
  \end{displaymath}
and for all $0\leq t\leq 1$ such that $p^-(t)\leq 1/2$,
 \begin{displaymath}
  \mathbb P \left ( \widehat \sigma^2_{\rm med} 
\leq  
   \beta \sigma^2-\sigma^2 \beta t \right ) \leq \exp \left( \frac{K_f}{2} \ln \Big [ 4p^-(t)(1-p^-(t)) \Big ] \right),
 \end{displaymath}
 where
\begin{displaymath}
p^+(t)=\mathbb P(\Gamma_{n_\infty} \geq n_\infty + 2n_\infty t) \quad \textrm{and} \quad p^-(t)=\mathbb P(\Gamma_{n_0} \leq n_0-n_0 t).
\end{displaymath}
\end{enumerate}
\end{prp}

Obviously, the Fourier support ${\rm supp}(\widehat f)$ and the subset $\overline{I_f}$ remain generally unknown in applications and have to be approximated. Let us recall that the main issue for estimating $\sigma$ comes from the bias term $\frac{\|f_{k}\|_2^2}{n_{k}}$ in Equation \eqref{eq-dirichletLike}, and in particular when the value $\sigma^2$ is negligible compared to $\frac{ \|f_{k}\|_2^2}{n_{k}}$. Therefore, a suitable candidate to approximate $\overline{I_f}$ will be some subset $\overline{J_f} \subset I_\L$ for which the impact of larger values $\frac{ \|f_{k}\|_2^2}{n_{k}}$ is minimized. This is made clear by Proposition \ref{prp-orderedTerms} below. The latter involves the following concentration bounds for Gaussian random variables: for all $0 < \alpha < 1$
\begin{equation}
\label{equ-concentrationBound} 
\mathbb P( |\langle f_{k}, \xi_{k} \rangle| \geq t_{\alpha,\sigma}\| f_{k} \|_2 )  \leq  \alpha \quad \textrm{where} \quad t_{\alpha, \sigma} = \sigma \times \sqrt{- 2\ln \left (\frac{\alpha}{4} \right )}.
\end{equation}

\begin{prp}
\label{prp-orderedTerms}
Let $0 < \alpha < 1$. Let $t_{\alpha, \sigma}$ be defined by Equation \eqref{equ-concentrationBound}. Assume that $f_{\ell} = 0$ and that the following inequality holds:
\begin{equation*}
\frac{\|f_k\|_2^2 + 2 t_{\alpha, \sigma} \|f_k\|_2}{\sigma^2}
\geq  \Phi_{ \frac{n_k}{n_\ell} \Gamma_{n_\ell} - \Gamma_{n_k}}^{-1} \left ( 1 - \frac{3\alpha}{2} \right ).
\end{equation*}
Then, the quantities
\begin{equation*}
b_k = \frac{\|\xi_{k}\|_2^2 + \|f_k\|_2^2 + 2 \langle \xi_k , f_k \rangle}{n_k} \quad \textrm{and} \quad b_\ell = \frac{\|\xi_\ell \|_2^2}{n_\ell}
\end{equation*}
satisfy $\mathbb P( b_k \geq b_\ell ) \geq 1 - \alpha$.
\end{prp}

By invariance under permutations, one may assume without loss of generality that the values $c_k$ are ordered in the decreasing order. Proposition~\ref{prp-orderedTerms} quantifies the fact that the highest values of $c_k$ correspond most likely to the indices $k$ for which $f_k \neq 0$. Consequently, setting $\overline{J_f}(r) = \sqcup_{k \in \{ r, r+1, ... K-r \}} I_k$ for all $1\leq r\leq \frac{K}{2}$, the estimators introduced in Equation \eqref{equ-simpleEstimators} may be rewritten replacing the unknown subset $\overline{I_f}$ by its known approximation $\overline{J_f}(r)$. So we define the estimators 
\begin{equation*}
\widehat \sigma_{\rm mean}^{r}(c)^2 = \frac{1}{| \{  k: I_k\subset \overline{J_f}(r)  \} |}\sum_{k: I_k\subset \overline{J_f}(r)} c_k \quad \textrm{and} \quad \widehat \sigma_{\rm med}^{r}(c)^2 = \mathrm{med}_{k: I_k \subset \overline{J_f}(r)}(c_k).
\end{equation*}

It is worth noting that from the symmetry of the subset $\overline J_f(r)$, it follows that the value $\widehat \sigma_{\rm med}^r$ actually does not depend on parameter $r$, and one will write $\widehat \sigma_{\rm med}$ in place of $\widehat \sigma_{\rm med}^r$. 

\subsection{Denoising Frequency Sparse Signals}

Let us begin with a result illustrating that localized Fourier analysis in $I_\L$ provides strong benefits in noise reduction tasks when the underlying signal is frequency sparse.

\begin{prp}\label{prp-denoiseTrivial}
Assume $f=f_{I}$ for some subset $I\subset I_{\L}$. Then 
\begin{displaymath}
\mathbb{E}\left [  \big \| f -\widetilde{f}_{I} \big \|_2^{2} \right ] = \mathbb{E} \left [ \big \| f -\widetilde{f} \big \|_2^{2} \right ] 
- \sigma^{2} \big \vert \overline{I} \cap {\rm sp} (\L) \big \vert.
\end{displaymath}
In particular, denoising of $\widetilde{f}$ boils down to denoising of $\widetilde{f}_{I}=f_{I}+\xi_{I}$.
\end{prp}

While Proposition~\ref{prp-denoiseTrivial} asserts a trivial denoising solution in the Fourier domain, \emph{i.e.} simply destroying the projection $\widetilde{f}_{\overline{I}}=\xi_{\overline{I}}$, this approach is no longer that immediate when considering the graph domain observations since the Fourier support of $f$ is unknown in practice and needs to be estimated. Based on the $\chi^2$-statistics, Algorithm~\ref{algo:algo1} is designed for this purpose. To the best of our knowledge, previous works that proposed method for Fourier support recovery for graph noisy signals~\cite{segarra2015aggregation} involve the complete eigendecomposition of matrix $\L$. The methodology suggested below makes use of projectors on eigenspaces which can be approximating with Chebyshev polynomials as detailed in the next Section~\ref{sec:chebyshev}.

\begin{algorithm}
\KwData{noisy signal $\widetilde f$, a subdivision $I_1, I_2, \ldots, I_{K}$, estimated $n_k = |I_k \cap \mathrm{sp}{L}|$, $k=1, \ldots, K$, threshold $\alpha \in (0,1)$}

\KwResult{$\widetilde f_{I}=P_{I}(\L)\widetilde f$, where $I$ is an approximation of the Fourier support of $\widetilde f$}

for $k=1, \ldots, K$\\

 Compute $\|\widetilde f_{k}\|^2_{2}= \| P_{I_{k}}(\L) \widetilde f \|^2_{2}$;\\
 Compute
 \[
  p_k=\mathbb P(\sigma^2 \Gamma_{n_k} > \|\widetilde f_{k}\|^2_2) \quad \textrm{and} \quad \Gamma_{n_k} \sim \chi^2(n_k);
\]

Compute $\widetilde f_I = \displaystyle\sum_{ k:~p_k \leq \alpha} P_{I_k}\widetilde f$.
\caption{Support approximation in the Fourier domain for noisy signal}
\label{algo:algo1}
\end{algorithm}

Heuristically, if $I$ contains the support of the Fourier transform of $f$, on the complementary subset $\overline I$ we only observe pure white Gaussian noise so that $\| P_{\overline I}\widetilde f \|_2^2=\| \widetilde f_{\overline I} \|_2^2$ is distributed as $\sigma^2 \chi^2(n_I)$ with $n_I=|\overline I \cap \mathrm{sp}(L)|$. On the other hand, on $I$ the square of the Euclidean norm of a non-centered Gaussian vector is observed. Consequently, the quantity $\mathbb P \left (\chi^2(n_I) > \sigma^{-2} \| P_{I} \widetilde f \|^2_2 \right )$ is typically very close to zero whereas $\mathbb P \left (\chi^2(n-n_I) > \sigma^{-2} \| P_{\overline I} \widetilde f \|^2_2 \right)$ remains away from $0$. To put it in a nutshell, sliding a window along the spectrum of $\L$, Algorithm~\ref{algo:algo1} performs a series of $\chi^2$-test. 

With the objective to provide theoretical guarantees that $\chi^2$-tests approach ${\rm supp}(\widehat f)$ correctly, it is important to turn the condition on the $p_k$-value into a condition involving only the values $\|f_k\|_2$ and $\sigma$. The next lemma shows that for sufficiently large values of the ratio $\frac{\|f_k\|_2}{\sigma}$, the inequality $p_k\leq \alpha$ holds so that the corresponding components ${\rm supp}( \widehat f_k )$ of the Fourier domain are legitimately included in the support estimate $I$.

\begin{lem}
\label{lem-pk<alpha}
Let $0<\alpha<1$ and let $\Gamma_{n_k} , \Gamma_{n_k}^\prime$ be two \emph{i.i.d} $\chi^2(n_k)$ random variables. Assume that:
\begin{displaymath}
\frac{\|f_k \|_2}{\sigma} \left( \frac{\|f_k\|_2}{\sigma} - 2\frac{t_{\alpha / 2,\sigma}}{\sigma} \right) \geq \Phi_{\Gamma_{n_k} - \Gamma_{n_k}^\prime}^{-1} \left ( 1- \frac{\alpha}{2} \right ),
\end{displaymath}
where $t_{\alpha,\sigma}$ is defined by Equation \eqref{equ-concentrationBound}. Then $p_k \leq \alpha$.
\end{lem}

In contrast to Lemma \ref{lem-pk<alpha}, the lemma below states that condition $p_k > \alpha$ holds for sufficiently small values of ratio $\sigma^{-1}\|f_k\|_2$.
\begin{lem}
\label{lem-pk>alpha}
Let $0<\alpha<1$ and let $\Gamma_{n_k}$ be a $\chi^2(n_k)$ random variable. For $0<\beta<1$, set $t_{\beta, k} = \sigma^2 \Phi_{\Gamma_{n_k}}^{-1}(1-\beta)$. Assume that
\begin{displaymath}
\left ( \frac{\|f_k\|_2+\sqrt{t_{\beta, k}}}{\sigma} \right )^2 < \Phi_{\Gamma_{n_k}}^{-1} \left ( 1 - \frac{\alpha}{1-\beta} \right ).
\end{displaymath}
Then $p_k > \alpha$.
\end{lem}

Compared to Proposition~\ref{prp-denoiseTrivial}, the result below quantifies the error resulting by approximating the support running Algorithm~\ref{algo:algo1}. Note that the requirement to have a constant sequence $(n_k)_k$ is used for statement clarity but similar assertions hold for the case $n_k \neq n_{k'}$.

\begin{prp}
\label{prp-denois-unknownI}
Set $f_I = \sum_{k: p_k \leq \alpha} P_{I_k}f$. Assume that $n_k = n_1$ for all $1\leq k\leq K$. Then,
\begin{enumerate}
\item the Fourier support approximation $\ell_2$-error satisfies 
\begin{equation}
\label{equ-FourierSupportError}
 \|f-f_I\|_2^2 \leq |\{ k, I_k\subset I_f, p_k > \alpha \}| \left( t_{\alpha / 2, \sigma}  +  \sqrt{ t_{\alpha / 2, \sigma}^2  +  \left( \sigma \Phi_{\Gamma_{n_1} - \Gamma_{n_1}^{\prime}}^{-1} \left ( 1- \frac{\alpha}{2} \right ) \right)^2} \right)^2.
 \end{equation}
\item the Noise $\ell_2$-error on Fourier support:
\begin{equation}
\label{equ-FourierNoiseError}
\mathbb E \|f_I - \widetilde f_I \|_2^2  =  |\{ k, p_k \leq \alpha \}| n_1 \sigma^2 .
\end{equation}
\end{enumerate}
\end{prp}

 Lemma \ref{lem-pk>alpha} asserts that the set $\{k, p_k > \alpha \}$ is small when most of the values $\|f_k\|_2$ are large enough compared to noise level $\sigma$ for $I_k \cap {\rm supp}\widehat f \neq \emptyset$. In such a case, Fourier support approximation $\ell_2$-error is small. Regarding the noise $\ell_2$-error, the inclusion $ \{k, p_k \leq \alpha\} \subset \{k, I_k\cap{\rm supp}\widehat f \neq \emptyset \}$ holds by Lemma \ref{lem-pk>alpha}. Moreover, Lemma \ref{lem-pk<alpha} asserts that the set $\{ k, p_k \leq \alpha\}$ contains the entire set $\{k, I_k \cap {\rm supp}\widehat f \neq \emptyset\}$ for sufficiently large values of $\sigma^{-1}\|f_k\|_2$ when $I_k \cap{\rm supp}\widehat f \neq \emptyset$. For such favorable situations, the noise $\ell_2$-error is exactly $ n_1 \sigma|\{k, I_k \subset I_f\}|$, the amount of noise on the extended support $I_f$. 

\begin{algorithm}
\KwData{$\widetilde{f}$, $\alpha$, $(I_k)_{k=1, \ldots, K}$, estimated $n_k = |I_k \cap \mathrm{sp}(\L)|$, thresholds $t_1$, $t_2$}
\KwResult{estimator $\widehat{f}$ of signal $f$}
 Apply Algorithm~\ref{algo:algo1} with $\widetilde{f}$, $\alpha$, $(I_k)_{k=1, \ldots K}$, estimated $n_k$; it outputs $\widetilde{f_{I}}$ and $\widetilde{f_{\overline I}}$;\\
 Apply soft-thresholding with threshold $t_1$ to $\mathcal W^{I}\widetilde f$ and $t_2$ to $\mathcal W^{\overline I} \widetilde f$;\\
 Apply the inverse LocLet transform to the soft-thresholded coefficients to obtain $\widehat{f_I},\widehat f_{\widetilde I}$; \\
 Compute the estimator $\widehat{f} = \widehat f_{I} + \widehat f_{\overline{I}}$;
\caption{LocLets thresholding estimation procedure}
\label{algo:algo2}
\end{algorithm}

The second step gives an estimate of the original signal using a thresholding procedure on each element $\widetilde f_I$ and $\widetilde f_{\overline I}$. On the one hand, the methodology developed in \cite{gobel2018construction} is prohibitive in terms of time and space complexity as soon as the underlying graphs become moderately large. On the other hand, the fast SGWT remains an approximating procedure. If a signal happens to be very frequency-sparse, then an even more optimal strategy is possible: first, the support $I$ in the frequency domain is approximated with the help of Algorithm \ref{algo:algo1}; then, the procedure of \cite{gobel2018construction} is applied to $P_If$ (the low-rank part) and LocLets on $P_{\overline I}(\L)f$. This idea is made precisely in Algorithm \ref{algo:algo3}.

\begin{algorithm}
\KwData{$\widetilde{f}$, $\alpha$, $(I_k)_{k=1, \ldots, K}$, estimated $n_k = |I_k \cap \mathrm{sp}(\L)|$, thresholds $t_1$, $t_2$}
\KwResult{estimator $\widehat{f}$ of signal $f$}
 Apply Algorithm~\ref{algo:algo1} with $\widetilde{f}$, $\alpha$, $(I_k)_{k=1, \ldots K}$, estimated $n_k$; it outputs $\widetilde{f_{I}}$ and $\widetilde{f_{\overline I}}$;\\
 Compute Parseval Frame for $\L_I$;\\
 Apply Parseval Frame thresholding with threshold $t_1$ to $\widetilde{f}_I$; it outputs $\widehat{f_I}$;\\
 Apply soft-thresholding with threshold $t_2$ to $\mathcal W^{\overline I} \widetilde f$;\\
 Apply the inverse LocLet transform to the soft-thresholded coefficients to obtain $\widehat f_{\widetilde I}$; \\
 Compute the estimator $\widehat{f} = \widehat f_{I} + \widehat f_{\overline{I}}$;
\caption{LocLets support approximation, and low-rank Parseval Frame thresholding procedure}
\label{algo:algo3}
\end{algorithm}

Estimator $\widehat f$ produced in Algorithm~\ref{algo:algo3} satisfies a tighter oracle bound inequality than the one given in  \cite[Theorem 3]{gobel2018construction}. This theoretical guarantee is widely supported by our experiments described in Section \ref{experiments}. Following notations from \cite[Equation (21)]{gobel2018construction}, we denote by $OB(f_I)$ the oracle bound obtained from an \textit{oracle} estimator of $f_I$ from a noisy $\widetilde f_I$ exploiting some knowledge about the unknown signal $f_I$. We refer to~\cite{gobel2018construction} for precise details.

\begin{thm}
\label{thm-oracleBound}
Let $I, \widehat f$ be respectively the support approximation and the estimator of $f$ obtained from Algorithms \ref{algo:algo1} and \ref{algo:algo3} with threshold value $t_2 = 0$. Then we have
\begin{displaymath}
\mathbb E  \|f-\widehat f \|_2^2 \leq \mathbb E  \|f-f_I \|_2^2+(2 \log(n_I) + 1)(\sigma^2 + OB(f_I)).
\end{displaymath}
\end{thm}

The right-hand side in the inequality of Theorem \ref{thm-oracleBound} has a more explicit expression in terms of $\alpha, \sigma$ using Proposition \ref{prp-denois-unknownI}. Up to the error made by approximating the support with Algorithm \ref{algo:algo1}, the $\ell_2$-risk is essentially bounded by the $\ell_2$-risk of the Parseval frame procedure from~\cite{gobel2018construction} on the low-rank projection $f_I$ of $f$, that is
\begin{displaymath}
\mathbb E  \|f-\widehat f\|_2^2 \lesssim (2 \log(n_I) + 1) (\sigma^2 + OB(f_I)).
\end{displaymath}
To conclude, Theorem \ref{thm-oracleBound} provides a theoretical guarantee that the support approximation improves the denoising performances obtained from \cite{gobel2018construction}.

\section{Properties of LocLets}
\label{sec:chebyshev}

In this section, we highlight important properties for the application of Fourier localization in practice. First we discuss computational analysis, and methods to apply our techniques to large graphs. Then we study the relationships of LocLets with well-known graph wavelet constructions.

\subsection{Fast LocLet Transform and Computational Analysis}

In the case of large graphs,  GSP requires a special care for being efficient since functional calculus relies \emph{a priori} on the complete reduction of the Laplacian. Actually, several efficient methods were designed to retrieve only partial information from the eigendecomposition as matrix reduction techniques (see for instance~\cite{mahadevan2008fast, susnjara2015accelerated}) or polynomial approximations~\cite{hammond2011wavelets, shuman2011chebyshev,di2016efficient}. In this paper, the widely adopted latter approach with Chebyshev polynomials approximation is preferred and briefly recalled below (we refer the reader to \cite[Section III.C.]{shuman2011chebyshev} for a brief but more detailed description of Chebyshev approximation).

\subsubsection{Chebyshev approximations}

Roughly speaking, the idea is to approximate the function $g$ with its Chebyshev expansion $g_N$ at order $N$. More precisely, the Chebyshev polynomials of the first kind $(T_i)_{i \geq 0}$ are defined from the second order recursion
\begin{equation*}
T_0(x) = 1, \quad T_1(x) = x, \quad |x|\leq 1, \quad \textrm{and} \quad T_i(x) = xT_{i-1}(x) - T_{i-2}(x),
\end{equation*}
for $i \geq 2$. Then, the matrix $\L$ is \emph{normalized} as $\widetilde{\L} = \frac{2}{\lambda_1}\L - I_n$ so that $\mathrm{sp}(\widetilde{\L}) \subset [-1, 1]$. This gives rise to some function $\widetilde g : [-1, 1]\rightarrow \R$ with the property $g(\L) = \widetilde g(\widetilde \L)$. In fact, $\widetilde g(x) = g(\frac{\lambda_1}{2}(x+1))$ for all $x \in [-1, 1]$. Then $g(\L)$ has the following truncated Chebyshev expansion $g(\L) \approx g_N(\L)$:
\begin{equation*}
g_N(\L) = \sum_{0\leq i\leq N} a_i(\widetilde g) T_i(\widetilde \L),
\end{equation*}
where $N$ is the maximal degree of polynomials $T_i$ used in the expansion, and $a_i(\widetilde g)$ is the $i$-th coefficient in the $N$-th order Chebyshev expansion of function $\widetilde g$. Following \cite{hammond2011wavelets}, for any filter $g$ on $\spL$ and any signal $f$ on graph $\G$, the approximation $g_N(\L)$ provides a vector value close to $g(\L)f$ with time complexity $O(|\mathcal E|  N)$. 

The object presented in the sequel involves in particular the spectral projection $P_I(\L)f$ of a signal $f$ for any subset $I\subset I_\L$ which can be derived from the Chebyshev expansion of the indicator function $g = \indic_I$. This observation actually appears in several recent works~\cite{di2016efficient, fan2019spectrum}. More importantly for our study,  this efficient estimation is part of the Hutchinson stochastic trace estimator technique \cite{hutchinson1990stochastic}, providing us with an effective method to estimate $n_I = \rm{Tr}(\L_I)$. Finally, the present paper focuses on the computation of a sequence $g(\L_{I_k})_{1\leq k\leq K}$ (or its vector counterpart $g(\L_{I_k})f$) instead of a single $g(\L)$ (resp. $g(\L)f$). While a naive estimation would suggest that the computational complexity is then multiplied by a factor $K$ compared to the complexity of the computation of $g(\L)$, we argue in the following that there is in fact no significant computational overhead.

\subsubsection{Sharing Chebyshev polynomial matrices among filters}


Let us assume that it is needed to compute the estimated values of $g_k(\L)f$ for a given signal $f$ for several filters $g_k$, $k=0, \ldots, K$. Then the following two-step strategy can be adopted: (1) pre-compute Chebyshev expansions $\widetilde g_k(x) \approx \widetilde g_{k, N}(x) = \sum_{0\leq i\leq N} a_i(\widetilde{g_k}) T_i(x)$ for all $k=0, \ldots, K$; independently, compute Chebyshev approximation vectors $T_i(\widetilde \L)f$ for all $0\leq i\leq N$; (2) combine the previous results to compute the Chebyshev approximation $g_{k,N}(\L)f$ of $g_k(\L)f$:
\begin{equation*}
g_{k,N}(\L)f =  \sum_{0 \leq i \leq N} a_i(\widetilde g_k) T_i(\widetilde \L)f.
\end{equation*}

The complexity of the first step is dominated by the $N$ matrix-vector multiplications required to obtain $T_i(\widetilde \L)f$. So the first step has complexity $O(|\mathcal E|  N)$. The second step adds $N$ weighted matrices $a_i(\widetilde g_k) T_i(\widetilde \L)$ together, which is an operation of complexity $O(N  n^2)$ at most. As an important matter of fact, the overall complexity for this procedure is bounded by $O(|\mathcal E|  N  +  N  n^2)$, which is independent of the number of filters $g_k$, and the same as for the computation of $g(\L)$.

Sharing matrices among filters has several examples of applications in the current paper:
\begin{enumerate}
\item Computation of $g( \L_{I_k} )f$ for all $1\leq k \leq K$: the equation $g(\L_{I_k})f = g(\indic_{I_k}(\L) \L)f$ holds so that we can consider filters $g_k(x) = g(\indic_{I_k}(x) x)$. 
\item Computation of $g(s \L)f$ for several scale values $s$: consider filters of the form $g_s(x) = g(sx)$.
\item Computation of $n_{I_k}$ for all $1\leq k \leq K$: Hutchinson's stochastic estimation computes averages of $f_i^T P_{I_k}(\L)f_i$ for some random vectors $f_i$ ($i\leq n_H$) whose computational complexity is dominated by the approximation of vectors $P_{I_k}f_i$. Considering filters $g_k(x) = \indic_{I_k}(x)$, and sharing random vectors $(f_i)_i$ among all approximations of $n_k$, we end up with a complexity of $O(n_H  N  |\mathcal E|)$, independent of value $K$.
\end{enumerate}

In particular, Algorithm~\ref{algo:algo1} has complexity $O(n_H  N  |\mathcal E| + N   n^2)$. Indeed, its efficiency is calibrated on the computations of sequences $(\|\widetilde f_k\|_2)_{1\leq k\leq K}$ and $(n_k)_{1\leq k\leq K}$ whose computational analysis was discussed previously. It is worth observing that values $n_k$ do not depend on signal $f$ and should be estimated only once in the case where several signals $\widetilde f_1 , \widetilde f_2, \ldots$ are to be denoised.

%
%
%
%

\subsubsection{Optimizing storage of LocLets coefficients}

The storage of wavelet coefficients $(\mathcal W^{I_k}f)_{1\leq k\leq K}$ requires \emph{a priori} $K$ times the storage cost associated with the original transform $\mathcal W f$. When matrix reduction techniques are used to compute wavelets transform~\cite{susnjara2015accelerated}, one may reduce the storage consumption of the localized SGWT by suitably choosing the \textit{impulse} functions $(\delta_{m})_m$. For instance, assume that for each subset $I_k$ a Lanczos basis $(v_{m}^{k})_m$ of the subspace spanned by $\{ \chi_\ell, \ell \in I_k \}$ is given. Then the size of sequences $(v_{m}^{k})_m$ and $(v_{m}^{k})_{m, k}$ are respectively of order $O(|I_k \cap {\rm sp}(\L)|)= O(n_k)$ and $O(n)$. Thus, with impulse functions $(v_{m}^{k})_n$ in place of $\delta_m$ for transform $\mathcal W^{I_k}$, the storage requirements of localized transform $(\mathcal W^{I_k}f)_{1\leq k\leq K}$ and the original one $\mathcal W f$ are of the same order $O(J  n)$.

\subsection{Connections with Well-know Frames}

A family $\mathfrak F=\{ r_i \}_{i \in I}$ of vectors of $\mathbb R^\mathcal{V}$ is a frame if there exist $A,B > 0$ satisfying for all $f \in \mathbb R^\mathcal{V}$
\begin{equation*} 
A\|f \|^2_2 \leq \sum_{i \in I} |\langle f,r_i \rangle|^2 \leq B \|f\|^2_2.  
\end{equation*}
A frame is said to be \textit{tight} if $A=B$. This section gives two examples of frames introduced in the literature which can be realized as a LocLets representation and thus benefit from the advantages given by localization in the spectrum.

\subsubsection{Parseval frames}
\label{sec-ParsevalFrame}

Parseval frames are powerful representations to design wavelets with nice reconstruction properties \cite{leonardi2013tight,gobel2018construction}. In this section, we investigate the extent to which Parseval frames can be obtained from some LocLet representation. We show that for a particular choice of partition $I_\L = \sqcup_k I_k$, there exist frames which are Parseval frames and composed only of LocLets functions.

A finite collection $(\psi_j)_{j=0, \ldots,J}$ is a finite partition of unity on the compact $[0,\lambda_1]$ if
\begin{equation}
\label{equ-parsevalFrame}
\psi_j : [0,\lambda_1] \rightarrow [0,1] \quad \textrm{for all} \quad j \leq J \quad \textrm{and} \quad \forall \lambda \in [0,\lambda_1], \quad \sum_{j=0}^J \psi_j(\lambda)=1.
\end{equation}
Given a finite partition of unity $(\psi_j)_{j=0, \ldots, J}$, the Parseval identity implies that the following set of vectors is a tight frame:
\begin{equation*}
\mathfrak F = \left \{ \sqrt{\psi_j}(\L)\delta_i, \quad j=0, \ldots, J, \quad i \in V \right \}.
\end{equation*}
Some constructions of partition of unity involve functions $(\psi_j)_j$ that have almost pairwise disjoint supports \emph{i.e.} ${\rm supp}(\psi_{j})\cap {\rm supp}(\psi_{j^\prime}) = \emptyset$ as soon as $|j - j^\prime|>1$. For such partition of unity, set $I_0 = {\rm supp}(\psi_{0})$, $I_J = I_0 = {\rm supp}(\psi_{J})$ and $I_j = {\rm supp}(\psi_{j}) \cap {\rm supp}(\psi_{j+1})$ for all $1\leq j\leq J-1$. Then, the sequence $(I_j)_{0\leq j\leq J}$ defines a finite partition of $[0, \lambda_1]$, $[0, \lambda_1] = \sqcup_{0\leq j\leq J} I_j$, such that:
\begin{equation}
\label{equ-correspondence}
\psi_0 \indic_{I_0} = \indic_{I_0},\quad (\psi_j + \psi_{j+1}) \indic_{I_j} = \indic_{I_j}, \quad 0< j < J, \quad \textrm{and} \quad \psi_J \indic_{I_J} = \indic_{I_J}.
\end{equation}
An alternative tight frame can be constructed using a LocLet representation as shown in the following proposition.

\begin{prp}
\label{prp-tightFrames}
Assume Equations~\eqref{equ-parsevalFrame} and~\eqref{equ-correspondence} hold and set, for all $1 \leq k \leq J$, $\varphi_{n,k}=\sqrt{\psi_{0}}(\L_{I_k})\delta_n$, $\psi_{1,n,k}=\sqrt{\psi_{k}}(\L_{I_k})\delta_n$ and $\psi_{2,n,k}=\sqrt{\psi_{k}}(\L_{I_{k+1}})\delta_n$ for all $1\leq k\leq J$. Then $(\varphi_{n,k}, \psi_{j,n,k})_{1\leq j\leq 2,~1\leq m\leq n,~1\leq k\leq J}$ is a tight frame.
\end{prp}

The resulting tight frame of Proposition~\ref{prp-tightFrames} is actually frame of LocLets if additionally the functions $\psi_j$ is of the form $\psi_j = \psi_1(s_j .)$ for some scale parameter $s_j$, $1 \leq j \leq J$. This is typically the case for the frames introduced in \cite{leonardi2013tight,gobel2018construction}. In these papers, the partition of unity is defined as follows: let $\omega : \mathbb R^+ \rightarrow [0,1]$ be some function with support in $[0,1]$, satisfying $\omega \equiv 1$ on $[0,b^{-1}]$ and set $\psi_0(\cdot)=\omega(\cdot)$ and for $j=1, \ldots, J$
\begin{equation*}
\psi_j(\cdot)=\omega(b^{-j} \cdot )-\omega(b^{-j+1} \cdot) \quad \textrm{with} \quad J= \left \lfloor \frac{\log \lambda_1}{\log b} \right \rfloor + 2.
\end{equation*}
In particular, the functions $\psi_k$ have supports in intervals $J_{k}=[b^{k-2}, b^{k}]$. Thus, one may define disjoint intervals $(I_{k})_{k}$ as follows: $I_{k} = [b^{k-1}, b^{k}]$. We have $J_{k} = I_{k} \cup I_{k+1} $, so that Equations \eqref{equ-correspondence} hold whereas the scaling property $\psi_j = \psi_1(b^{-1} .)$ is straightforward. By Proposition \ref{prp-tightFrames}, the set of vectors 
\begin{equation*}
\left \{ \sqrt{\psi_{0}}(\L_{I_k})\delta_n, \sqrt{\psi_{1}}(s_k\L_{I_k})\delta_n, \sqrt{\psi_{1}}(s_k\L_{I_{k+1}})\delta_n, \quad n,k \right \}
\end{equation*}
is a tight frame of LocLets. Observe that the transform $(\mathcal W^{I_k})_{I_k}$each component $\mathcal W^{I_k}$ of the LocLet transform $(\mathcal W^{I_k})_{I_k}$ only admit two scale parameters $s_k, s_{k-1}$.

\subsubsection{Spectrum-adapted tight frames}

Let us consider another family of tight frames tailored to the distribution of the Laplacian $\L$ eigenvalues proposed in \cite{shuman2015spectrum}. As shown below, these frames can be written in terms of a warped version of LocLets, and up to some approximation, in terms of (non-warped) LocLets. First, let us briefly recall the construction from \cite{shuman2015spectrum}.

The notion of \emph{warped} SGWT is introduced in~\cite{shuman2015spectrum} to adapt the kernel to the spectral distribution. Given a warping function $\omega: I_\L \rightarrow \R$, the warped SGWT is defined as:
\begin{equation*}
\mathcal W^{\omega} f = ( \varphi(\omega(\L) )f^{T} , \psi(s_{1}\omega(\L) )f^{T}, \ldots, \psi(s_J \omega(\L))f^T)^{T}.
\end{equation*}
As for our spectral localization, the objective of warping is to take benefits from the distribution of ${\rm sp}(\L)$ along interval $I_\L$. While the two techniques show similarities (\emph{e.g.} estimation of ${\rm sp}(\L)$ distribution), they are meant to answer different problems: warped SGWT is a technique to adapt the whole spectrum to some task (\emph{e.g.} producing a tight frame), whereas localized SGWT is designed to answer problems related to localized subsets in the spectrum (\emph{e.g.} denoising a frequency sparse signal). Here we show that the advantages of both LocLets and warped SGWT are obtained when the two methods are combined in a warped LocLet representation.

Let $\omega$ be some warping function on $I_\L$ chosen in the form $\omega(\cdot) = \log(C\omega_0(\cdot))$ where $\omega_0$ stands for the cumulative spectral distribution of $\L$ and $C$ is some normalization constant as shown in \cite{shuman2015spectrum}. Then, let $\gamma>0$ be an upper bound on $\rm{sp }(\L)$ and let $R,J$ be two integers such that $2 \leq R \leq J$. Setting $\omega_{\gamma, J, R} = \frac{\gamma}{J+1+R}$, Corollary 2 in~\cite{shuman2015spectrum} asserts that the family $(g_{m,j})_{m,j}$ of functions defined below is a tight frame 
\begin{equation}
\label{eq-spectrumAdaptedTightFrame}
g_{m,j}  =  \sum_{\ell}  \widehat g_j(\lambda_\ell) \widehat \delta_m(\ell) \chi_\ell,
\end{equation}
where functions $\widehat g_j$ arise from some kernel $\widehat g$ as
\begin{displaymath}
\widehat g_j(\lambda)= \widehat g(\omega(\lambda) - j\omega_{\gamma,J,R} )= \widehat g \left ( \log \frac{C\omega(\lambda)}{ e^{j\omega_{\gamma,J,R}} } \right ).
\end{displaymath}
Typically in \cite{shuman2015spectrum}, the kernel $\widehat g$ takes the form
\begin{displaymath}
\widehat g(\lambda)=\left [ \sum_{0\leq j\leq J} a_j \cos \left ( 2\pi j \cos \left ( \frac{\lambda}{R \omega_{\gamma,J,R}}  +  \frac{1}{2} \right ) \right )  \right ]\indic_{[ -R\omega_{\gamma,J,R}, 0]} (\lambda).
\end{displaymath}
for some sequence $(a_j)_j$ satisfying $\sum_j (-1)^j a_j = 0$.

The following proposition states that Equation~\eqref{eq-spectrumAdaptedTightFrame} admits an alternative form involving only (warped) LocLets functions. 

\begin{prp}
\label{prp-spectrumAdaptedFrames}
Setting $\psi(\lambda)=\widehat g(\log(C \lambda) )$ for $\lambda > 0$, consider the family of warped LocLets defined  for all $0 \leq k \leq R-1$, $1\leq m\leq n$ and $1\leq j\leq J$ by
\begin{equation*}
\psi_{j,m,I_k}  =  \sum_{\ell \in I_k}  \psi(s_j \omega_0(\lambda_\ell)) \widehat \delta_m(\ell) \chi_\ell \quad \textrm{with} \quad I_k = \left [ \frac{e^{(k-R)\omega_{\gamma,J,R}}}{C},\frac{e^{(k-R+1)\omega_{\gamma,J,R}}}{C} \right ].   
\end{equation*}
Then, the following identity holds for all $j=1, \ldots, J$ and all $m=1, \ldots, n$
\begin{equation*}
g_{m,j} = \sum_{1\leq k\leq R-1} \psi_{j,m,I_k}.
\end{equation*}
\end{prp}

%
%

%


\section{Experiments on suites of large matrices}
\label{experiments}

This section details experiments made on large graphs to validate the Fourier localization techniques introduced in that paper. After describing the experimental settings, we describe the outcomes of several experiments showing strong advantages in the use of Fourier localization in practice.

\subsection{Choice of spectral partition \texorpdfstring{$I_\L = \sqcup_k I_k$}{IL}}
\label{sec-Kselection}

In order to keep the problem combinatorially tractable, it is necessary to reduce the choice of possible partitions of $I_\L$ into subintervals $I_k$. That is why, the partitions considered in the sequel are regular in the sense that all intervals have the same length $\lambda_1/K$ for some integer $K \geq 1$. Thereafter, the parameter $K$ is chosen so that the eigenvalues are distributed as evenly as possible in each interval $I_k$. Without prior information, it is indeed natural not to favor one part of the spectrum over another. Most importantly, in the view of the concentration property of the median around the noise level $\sigma^2$ of Proposition~\ref{concentration-noiseLevel}, it is essential to keep the parameter $\beta$ as close to one as possible.   

In order to implement the ideas above, it is necessary to estimate the spectral measure of $\L$ which can be described by the so-called spectral density function:
\begin{equation*}
\varphi_\L(\lambda)=\frac{1}{n} \sum_{\ell=1}^n \delta(\lambda-\lambda_\ell) \quad \textrm{for all} \quad \lambda\in I_\L. 
\end{equation*}
There are several techniques for such an approximation among which the \emph{Kernel Polynomial Method} (see, \emph{e.g.} \cite{silver1994densities, wang1994calculating}). The latter approximates the spectral density $\varphi_\L$ with the help of matrix Chebyshev expansion $(\varphi_\L^N)_N$ (see~\cite{lin2016approximating} for a detailed presentation).

Now, let $(I_k)_{1 \leq k \leq K}$ be some regular partition of $I_\L$ and $(n_k)_{1 \leq k \leq K}$ be the corresponding numbers of eigenvalues in each $I_k$. Choosing the parameter $K \geq 1$ so that the entropy defined by
\begin{equation*}
E(K) = - \sum_{1\leq k\leq K} \frac{n_k}{n} \log\left(\frac{n_k}{n}\right)
\end{equation*}
is maximal ensures that the eigenvalues are as equally distributed in each interval as possible. In application, the Kernel Polynomial Method provides an approximation $n_k^N$ of $n_k$ and the corresponding empirical entropy $E_N(K)$ is used as a proxy for the theoretical one.

Empirically, the entropy increases logarithmically and then stabilizes from a certain elbow value $K_{\rm elbow}$ as illustrated in Figure~\ref{fig-entropy}. This elbow value is displayed in dashed lines in Figure~\ref{fig-entropy}. In the experiments, we choose this value $K_{\rm elbow}$ motivated by two reasons. First, as the intervals become shorter it is more difficult to obtain a uniform distribution of the eigenvalues into those intervals. The second reason is related  to the quality of the estimate $n_k^N$ of $n_k$ as the sample size decreases. To illustrate this fact, we consider the Mean Relative Error (MRE) defined by
\begin{equation*}
\mathrm{MRE}_N(K)=\frac{\sum_{1\leq k\leq K} | n_k - n_{k}^N |}{n}.
\end{equation*}
As highlighted by Figure~\ref{fig-entropy}, the empirical entropy actually stabilizes when the Chebyshev approximation, in terms of MRE, is no longer sharp enough. 

\begin{figure}[t]
  \centering
\subfigure[\textit{Si2} graph]{
\includegraphics[width=0.4\textwidth]{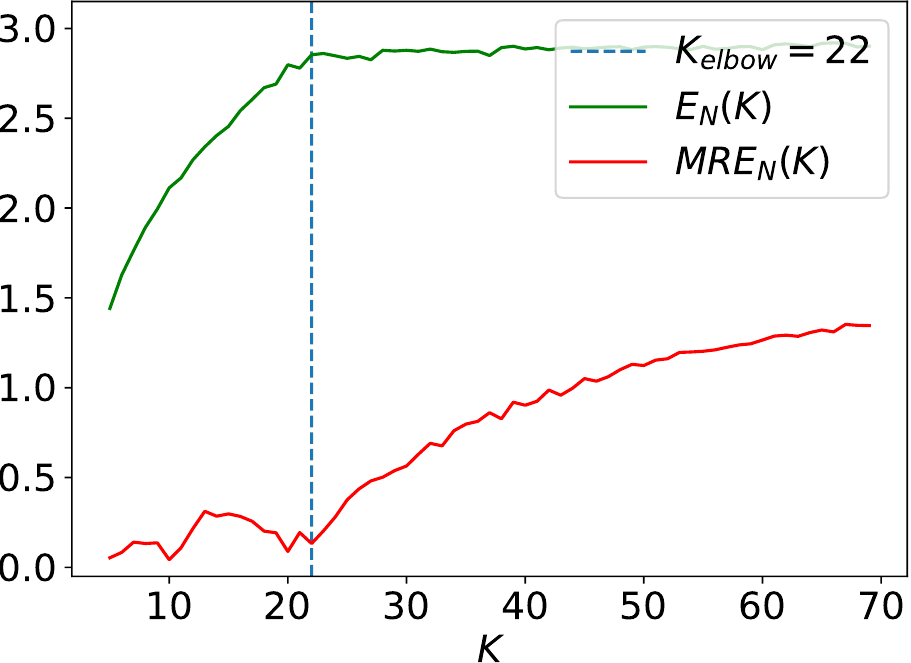}
}
~~~~
\subfigure[\textit{minnesota} graph]{
\includegraphics[width=0.4\textwidth]{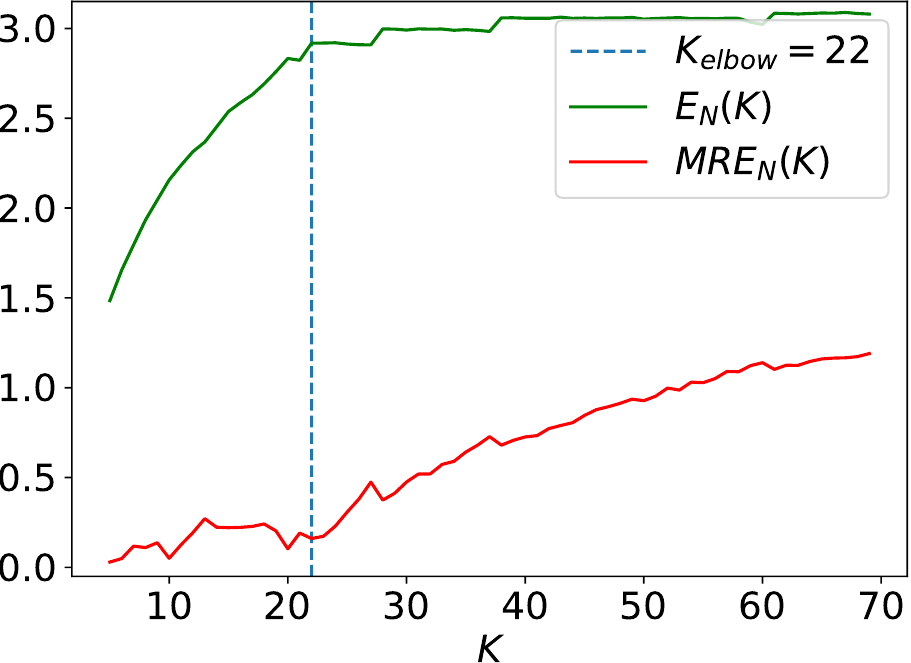}
}
\caption{Variations of $E_N(K)$ and $\mathrm{MRE}_N(K)$ with parameter $K$.}
\label{fig-entropy}
\end{figure}

\subsection{The experimental settings}


Following~\cite{fan2019spectrum}, we propose to validate our techniques on an extended suite of large matrices extracted from the Sparse Matrix Collection in~\cite{davis2011university}. Most of these matrices have an interpretation as the Laplacian matrix of a large graph. We define matrix $\L$ from the following matrices of the suite:\textit{si2} ($n=769$), \textit{minnesota} ($n=2642$), \textit{cage9} ($n=3534$), \textit{saylr4} ($n=3564$) and \textit{net25} ($n=9520$). We extend this graph collection with the well-studied \textit{swissroll} graph Laplacian matrix ($n=1000$).


We sample randomly signals whose supports are sparse in the Fourier domain. We will use the notation $f_{i-j}$ for normalized signals supported on a sub-interval of $I_\L$ containing exactly the eigenvalues $\lambda_{i}, \lambda_{i+1},\ldots,\lambda_{j}$. As an example, $f_{n-n}$ is a constant signal while $f_{1-2}$ is a highly non-smooth signal supported on the eigenspaces of large eigenvalues $\lambda_1, \lambda_2$. For experiments, the signals were calculated from the knowledge of $\mathrm{sp}(\L)$, and relevant projections of random functions on the graph.



We have compared the performances of Algorithms~\ref{algo:algo2} and~\ref{algo:algo3} against the thresholding procedure described in~\cite{gobel2018construction}. As the denoising method in \cite{gobel2018construction} requires the computation of the whole spectral decomposition of the Laplacian, it does not scale to large graphs. We stress here that we provide a fair comparison with \cite{gobel2018construction}, only in terms of denoising performance, and with no computational considerations. Moreover, we choose for LocLets to use the most naive thresholding procedure by considering a global and scale independent threshold level. 

For all the experiments below, the SGWT and LocLets are built upon the scale and kernel functions giving rise to the Parseval frame of ~\cite{gobel2018construction}, whose construction is recalled in Section~\ref{sec-ParsevalFrame}. More precisely, set respectively $\varphi=\sqrt{\zeta_0}$ and $\psi=\sqrt{\zeta_1}$ for the scale and kernel functions with $\zeta_0(x)=\omega(x)$, and $\zeta_1(x)=\omega(b^{-1}x)-\omega(x)$,
where we choose $b=2$ and $\omega$ is piecewise linear, vanishes on $[1,\infty)$ and is constant equal to one on $(-\infty,b]$. The scales are of the form $s_j=b^{-j+1}$ for $j=1, \ldots, J$ where $J$ is chosen similarly to \cite{gobel2018construction}. 

In what follows, `PF' stands for Parseval Frame and refers to the estimator of \cite{gobel2018construction}; the estimators implemented by Algorithm \ref{algo:algo2} and Algorithm \ref{algo:algo3} are referred to as `LLet' and `LLet+PF' respectively. The notation `$\mathrm{SNR}_{\mathrm{in}}$' refers to the trivial model releasing the noisy signal $\widetilde f$, corresponding to the classical input noise level measurement, and serves as a worst-case baseline for other models. Below, the latter methodology is shown to outperform all the others for very frequency-sparse signals. It is also worth recalling that `LLet+PF' benefits from the dimension reduction property of LocLets. More precisely, whereas the whole eigendecomposition of $\L$ is required to apply `PF', for Parseval frame denoising in the context of `LLet+PF', only a low-rank spectral decomposition is needed, namely the decomposition of $\L_I$ for $I$ the estimate of ${\rm supp}\widehat f$.  


For all our experiments, we set $\alpha=0.001$ for Algorithm \ref{algo:algo1}. For the denoising experiment, we compute the best SNR result $r_{D}$ over a large grid of values $(t_1, t_2)$, and  for each denoising method $D$ with $D \in\{`\mathrm{SNR}_{\mathrm{in}}\textrm{',`PF',`LLet',`LLet+PF'} \}$. Then, we calculate two metrics: the maximum $\mathrm{M}_{D}$ and average value $\mu_{D}$ of the values $r_{D}$ over $10$ random perturbations of the signal $f$. We recall that a good quality in denoising is reflected by a large value of the SNR metric.

\subsection{Analysis of our experiments}
\subsubsection{Noise level estimation}


We have evaluated the performances of estimators $\widehat \sigma_{\rm mean}^{r}$ and $\widehat \sigma_{\rm med}$ in the estimation of the unknown noise level $\sigma$ from $10$ realizations of the noisy signal $\widetilde  f = f + \xi$ for a given noise level $\sigma$. Figure~\ref{fig-sigmas} (resp. Figure~\ref{fig-sigma001Net25} ) shows the best performances of each estimator on the \textit{minnesota} (resp. \textit{net25}) graph for the non-regular but frequency sparse signal $f_{1392-1343}$ (resp. $f=f_{4971-5020}$), when parameter $K$ ranges in $\{5, 10, 20, 30, 40, 50 \}$ and for level of noise $\sigma=0.01$ (resp. $\sigma=0.001$).

\begin{figure}
  \centering
\includegraphics[width=0.4\textwidth]{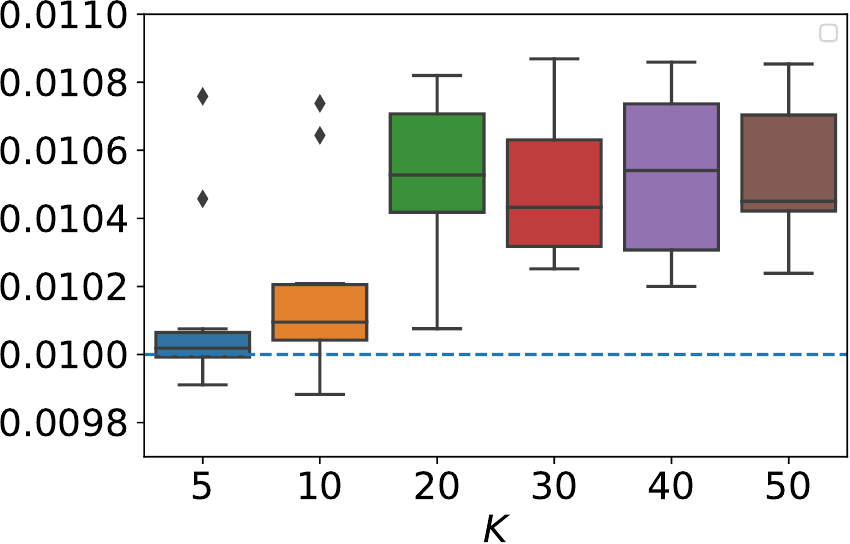}
\includegraphics[width=0.4\textwidth]{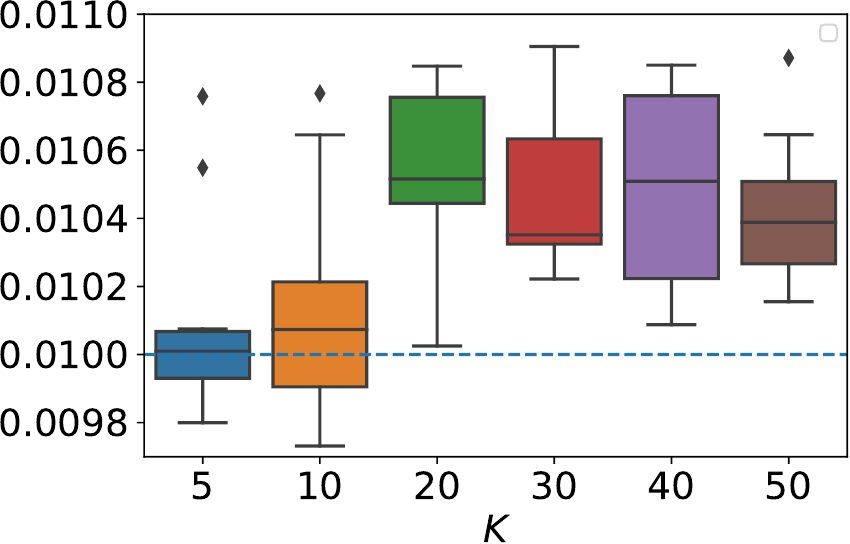}
\caption{Performances of estimators $\widehat \sigma_{\rm mean}$ (left) and $\widehat \sigma_{\rm med}$ (right) for \textit{minnesota} graph, signal $f_{1392-1343}$ and $\sigma=0.01$.} 
\label{fig-sigmas}
\end{figure}

Figure~\ref{fig-sigmas} illustrates that both estimators $\widehat \sigma_{\rm mean}^{r}$ and $\widehat \sigma_{\rm med}$ can provide good estimates of $\sigma$. Best performances are obtained for values of parameter $K$ below the elbow value $K_{elbow}(minnesota)=22$ introduced in Section~\ref{sec-Kselection}. We observe that performances drop considerably if almost no localization is used (for instance, for parameter values $K=1$ or $K=2$, $\widehat \sigma \sim 0.021$ in the experiment of Figure~\ref{fig-sigmas}, far from the performances for $K\geq 5$ for estimating $\sigma=0.01$).

\begin{figure}
\centering
  \includegraphics[width=0.4\textwidth]{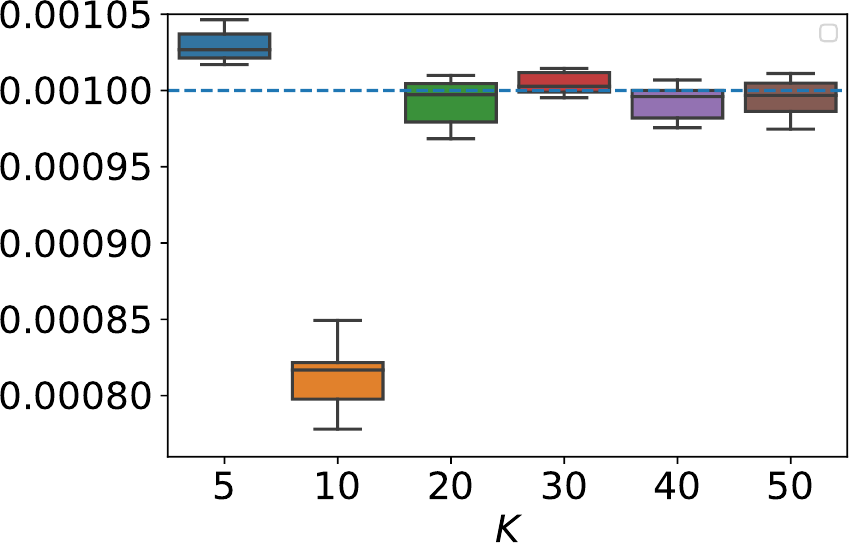}
\includegraphics[width=0.4\textwidth]{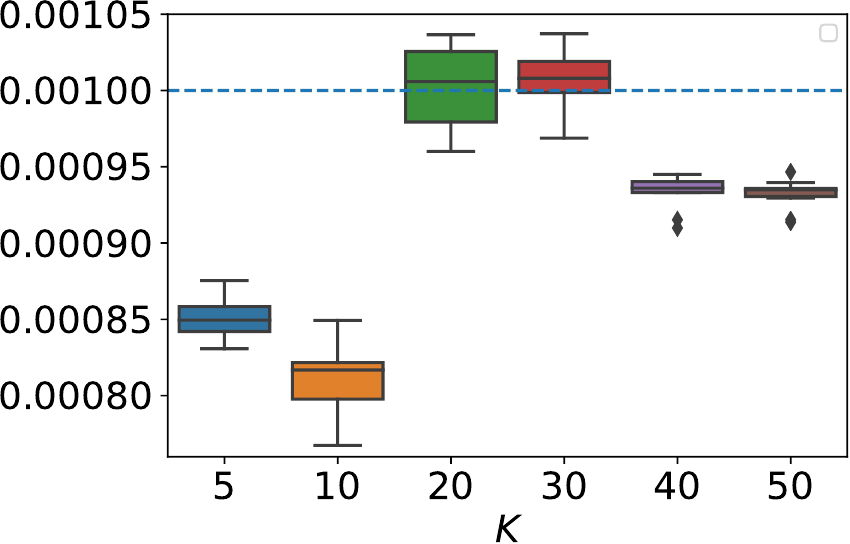}
\caption{Performances of estimators $\widehat \sigma_{\rm mean}^r$ (left) and $\widehat \sigma_{\rm med}^r$ (right) for \textit{net25} graph, signal $f_{4971-5020}$ and $\sigma = 0.001$.} 
\label{fig-sigma001Net25}
\end{figure}

Figure~\ref{fig-sigma001Net25} shows that localization is necessary, namely $K\geq 10$ or even $K\geq 20$, in order to reach the best performances for the large \textit{net25} graph. Contrary to experiments for the \textit{minnesota} graph, estimators  $\widehat \sigma_{\rm mean}^{r}$ and $\widehat \sigma_{\rm med}$ underestimate the value of $\sigma$. Also, best values of $K$ range between $10$ and $30$ for \textit{net25} graph, compared to best values $K=5$ and $K=10$ for \textit{minnesota} graph (see Figure~\ref{fig-sigmas}). This illustrates the idea that noise level estimation strongly depends on the underlying graph structure. As a consequence, a parameter $K$ selection has to take graph and signal information into account to be relevant. Interestingly, the elbow values $K_{elbow}(minnesota)=22$ and $K_{elbow}(net25)=22$ provide performances which are not optimal, but close to the best possible ones.

In Figure~\ref{fig-sigmasSi2r350}, performances for various values of parameter $r$ are displayed for a fixed parameter $K=K_{elbow}(minnesota)$. While it is true that $\widehat \sigma_{\rm mean}^{r}$ can perform better than $\widehat \sigma_{\rm med}^{r}$, it happens only for very specific values of $r$, which a priori depend on the signal regularity. Without any further parameter selection, these observations suggest using the most robust estimator $\widehat \sigma_{\rm med}$ in practice.


\begin{figure}[t]
\centering
\includegraphics[width=0.4\textwidth]{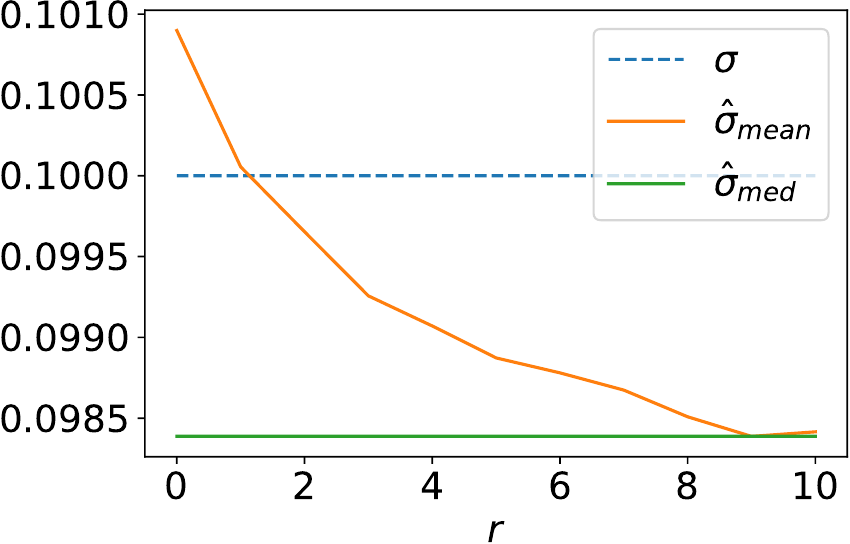}
\includegraphics[width=0.4\textwidth]{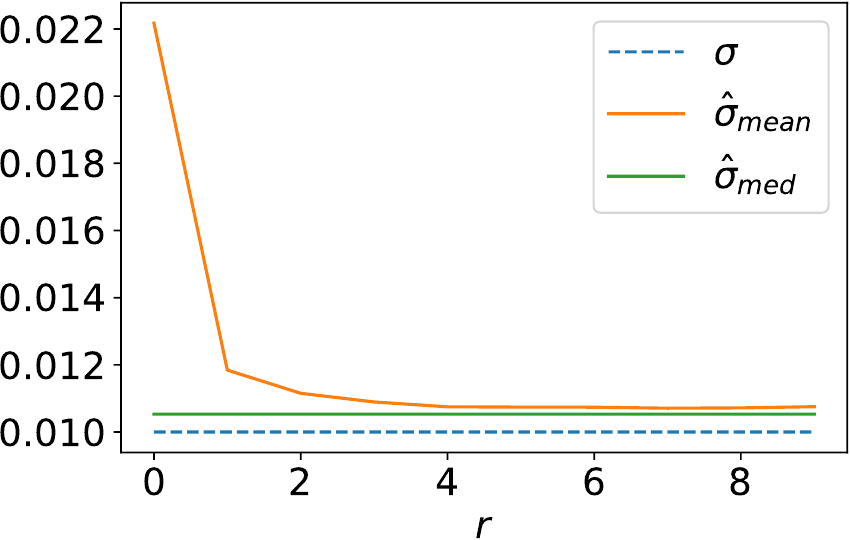}
\caption{Dependence on parameter $r$ for \textit{minnesota} graph, signal $f_{1392-1343}$, $K=22$ and $\sigma=0.1$ (left), $\sigma=0.01$ (right).} 
\label{fig-sigmasSi2r350}
\end{figure}


\subsubsection{Sparse signal denoising}

As a first denoising experiment, we have compared the performances of `LLet', `PF' and `LLet+PF' for a fixed value $K=K_{elbow}$ given by the rule of thumb described in Section~\ref{sec-Kselection}. For each matrix in the Extended Matrices Suite, we have experimented the denoising task on two frequency-sparse signals, one regular and the other non-regular. Several values of noise level $\sigma$ were used, corresponding to values of $\mathrm{SNR}_{\mathrm{in}}$ ranging in $[4, 18]$. Results from experiments are displayed in Tables~\ref{tab-swissroll} and~\ref{tab-all}. The first obvious observation is that `LLet+PF' performs better than its competitors in almost all situations. The gain is sometimes considerable since we observed a gap of $5$dB in $\mu_{D}$-metric between `LLet+PF' and its closest concurrent `PF' in some cases, and up to $7$dB in $\mathrm{M}_{D}$-metric. These experiments confirm the theoretical guarantees obtained in Theorem~\ref{thm-oracleBound}. The benefits of localization are reduced for graph \textit{net25}: Table~\ref{tab-all} shows that the more conservative choice $K=5$ is better than $K=25$. It appears that for \textit{net25}, the spectrum $\spL$ is localized at a small number of distinct eigenvalues, hence diminishing the advantages of localizing with our methods. 

\begin{table}[!htp] \label{tab-swissroll}
  \caption{SNR performance for Swissroll ($n=1000$, $K=22$).}
  \begin{center}
  \begin{tabular}{lrrrrrrrrr}
    \toprule
    signal  & $\sigma$ & $\mathrm{SNR}_{\mathrm{in}}$ & $\mathrm{M}_{\mathrm{PF}}$  &  $\mathrm{M}_{\mathrm{LLet}}$  &  $\mathrm{M}_{\mathrm{LLet+PF}}$ & $\mu_{\mathrm{PF}}$  &  $\mu_{\mathrm{LLet}}$  &  $\mu_{\mathrm{LLet+PF}}$ \\ 
    \midrule
    $f_{951-1000}$   & 0.005 & 16.195 &\ag 17.557 &\lg 20.580 &\gg 20.528 &\ag 17.361 &\lg 20.035 &\gg 19.974 \\
$f_{501-550}$ & 0.005 & 15.859 &\gg 18.298 &\ag 8.244  &\lg 20.821 &\gg 18.044 &\ag 8.140  &\lg 20.245 \\
$f_{951-1000}$   & 0.01  & 10.267  &\ag 12.183 &\gg 15.564 &\lg 15.701 &\ag 10.433 &\gg 13.652 &\lg 13.760 \\
$f_{501-550}$ & 0.01  & 10.178 &\gg 13.121 &\ag 7.879  &\lg 16.204 &\gg 11.165 &\ag 7.646  &\lg 14.518 \\
$f_{951-1000}$   & 0.015 & 6.763  &\ag 9.430  &\gg 12.661 &\lg 13.129 &\ag 8.961  &\gg 12.127 &\lg 12.388 \\
$f_{501-550}$ & 0.015 & 6.362  &\gg 9.898  &\ag 7.611  &\lg 14.159 &\gg 9.540  &\ag 7.481  &\lg 13.398 \\
    \bottomrule
  \end{tabular}
\end{center}
\end{table}

Another interesting observation is that `LLet' may outperform `PF' in some specific signal and noise level configurations, as shown in Table~\ref{tab-swissroll}. This is a very favorable result for localized Fourier analysis, since `LLet' appears to be a technique which is more accurate and more efficient as well compared to `PF' in some situations. However in many cases, `LLet' performances drop down compared to the more stable thresholding techniques `PF' and `LLet+PF', which use thresholds adapted to the wavelet basis.

\begin{table}[htp!]
\caption{SNR performances for denoising task.}
\label{tab-all}
\begin{center}
\begin{tabular}{lrrrrrrrrr}
\toprule
matrix  &  signal  & $\sigma$ & $\mathrm{SNR}_{\mathrm{in}}$ & $\mathrm{M}_{\mathrm{PF}}$  &   $\mathrm{M}_{\mathrm{LLet+PF}}$ & $\mu_{\mathrm{PF}}$  &  $\mu_{\mathrm{LLet+PF}}$ \\ 
\midrule
Si2 &$f_{720-769}$   & 0.005 & 17.104 &\ag 22.344 &\gg 26.973 &\ag 21.849 &\gg 25.170 \\
 ($n=762$, $K=22$)   &$f_{370-419}$  & 0.005 & 16.820 &\ag 18.034 &\gg 21.778 &\ag 17.813 &\gg 20.673 \\
&$f_{720-769}$   & 0.01  & 11.408 &\ag 16.821 &\gg 22.501 &\ag 16.572 &\gg 20.558 \\
&$f_{370-419}$  & 0.01  & 11.175 &\ag 12.444 &\gg 19.740 &\ag 12.151 &\gg 17.121 \\
&$f_{720-769}$   & 0.02  & 4.826  &\ag 11.476 &\gg 15.695 &\ag 11.104 &\gg 14.711 \\
&$f_{370-419}$  & 0.02  & 5.047  &\ag 7.354  &\gg 13.542 &\ag 6.925  &\gg 12.677 \\
\bottomrule
Minnesota &$f_{2593-2642}$    & 0.004 & 13.599 &\ag 17.839 &\gg 20.717 &\ag 17.672 &\gg 20.035 \\
($n=2642$, $K=22$)&$f_{1343-1392}$ & 0.004 & 13.741 &\ag 15.999 &\gg 20.388 &\ag 15.822 &\gg 19.234 \\
&$f_{2593-2642}$    & 0.005 & 11.681 &\ag 16.086 &\gg 19.342 &\ag 15.830 &\gg 18.417 \\
&$f_{1343-1392}$ & 0.005 & 11.916 &\ag 14.459 &\gg 18.392 &\ag 14.298 & \gg 18.029 \\
&$f_{2593-2642}$    & 0.01  & 5.911  &\ag 10.875 &\gg 14.143 &\ag 10.605 &\gg 13.556 \\
&$f_{1343-1392}$ & 0.01  & 5.843  &\ag 9.660  &\gg 11.762 &\ag 9.409  &\gg 10.952 \\
\bottomrule
Cage9 &$f_{3485-3534}$     & 0.003 & 15.016 &\gg 20.477 &\ag 9.700  &\gg 20.216 &\ag 9.664  \\
($n=3534$, $K=22$) &$f_{1785-1834}$ & 0.003 & 15.014 &\ag 15.876 &\gg 16.945 &\ag 15.798 &\gg 16.799 \\
&$f_{3485-3534}$     & 0.005 & 10.503 &\ag 17.290 &\gg 18.410 &\ag 16.772 &\gg 18.185 \\
&$f_{1785-1834}$ & 0.005 & 10.507 &\ag 12.118 &\gg 13.032 &\ag 12.035 &\gg 12.898 \\
&$f_{3485-3534}$     & 0.009 & 5.423  &\ag 13.002 &\gg 13.264 &\ag 12.763 &\gg 12.898 \\
&$f_{1785-1834}$ & 0.009 & 5.395  &\ag 8.329  &\gg 10.468 &\ag 8.168  &\gg 9.827  \\
\bottomrule
Saylr4 &$f_{3515-3564}$    & 0.003 & 14.871 &\ag 23.108 &\gg 24.516 &\ag 23.040 &\gg 24.117 \\
($n=3564$, $K=22$) &$f_{2015-2064}$ & 0.003 & 15.069 &\ag 21.365 &\gg 23.903 &\ag 21.010 &\gg 23.412 \\
&$f_{3515-3564}$    & 0.005 & 10.478 &\ag 19.135 &\gg 20.966 &\ag 18.943 &\gg 20.268 \\
&$f_{2015-2064}$ & 0.005 & 10.635 &\ag 17.016 &\gg 19.610 &\ag 16.662 &\gg 18.732 \\
&$f_{3515-3564}$    & 0.009 & 5.420  &\ag 15.277 &\gg 17.070 &\ag 14.791 &\gg 16.567 \\
&$f_{2015-2064}$ & 0.009 & 5.480  &\ag 12.055 &\gg 14.802 &\ag 11.830 &\gg 14.031 \\
\bottomrule
Net25 &$f_{9471-9520}$    & 0.006 & 4.682 &\ag 5.171 &\gg 5.319 &\ag 5.094 &\gg 5.205 \\
($n=9520$, $K=5$) &$f_{4971-5020}$ & 0.006 & 4.577 &\ag 5.811 &\gg 6.035 &\ag 5.714 &\gg 5.933 \\
&$f_{9471-9520}$    & 0.007 & 3.282 &\ag 5.287 &\gg 5.416 &\ag 5.127 &\gg 5.251 \\
&$f_{4971-5020}$ & 0.007 & 3.406 &\ag 5.267 &\gg 5.451 &\ag 5.104 &\gg 5.266 \\
&$f_{9471-9520}$    & 0.008 & 2.208 &\ag 5.171 &\gg 5.268 &\ag 4.928 &\gg 5.034 \\
&$f_{4971-5020}$ & 0.008 & 2.153 &\ag 4.450 &\gg 4.594 &\ag 4.319 &\gg 4.471 \\
\bottomrule
Net25 &$f_{9471-9520}$      & 0.006 & 4.495 &\ag 5.319 &\gg 6.004 &\ag 5.190 &\gg 5.834 \\
($n=9520$, $K=25$) &$f_{4971-5020}$ & 0.006 & 4.574 &\gg 5.849 &\ag 4.752 &\gg 5.714 &\ag 4.579 \\
&$f_{9471-9520}$    & 0.007 & 3.307 &\ag 5.264 &\gg 5.765 &\gg 4.951 &\ag 5.447 \\
&$f_{4971-5020}$ & 0.007 & 3.256 &\gg 5.132 &\ag 4.182 &\gg 5.014 &\ag 4.102 \\
&$f_{9471-9520}$    & 0.008 & 2.111 &\ag 4.913 &\gg 5.249 &\ag 4.756 &\gg 5.128 \\
&$f_{4971-5020}$ & 0.008 & 2.150 &\gg 4.345 &\ag 3.670 &\gg 4.251 &\ag 3.554 \\
\bottomrule
\end{tabular}
\end{center}
\end{table}


We also provide experimental results to understand the extent to which our results depend on the partition size parameter $K$. A few remarks are suggested by Figure~\ref{fig-KselectionSi2}: 
\begin{itemize}
\item The best performances are not obtained for the elbow value $K_{elbow}$, suggesting searching for a more task-adapted size of partition $K$.
\item Good performances persist for values of $K$ much larger than $K_{elbow}$, and in particular for regular signals.
\item For large values of $K$, there is a severe drop in performances. As explained before, the error generated by Chebyshev's approximation grows with the number of intervals in the partition, which makes the approximation of the support more difficult.
\end{itemize} 

\begin{figure}[t]
\begin{center}
\includegraphics[width=0.4\textwidth]{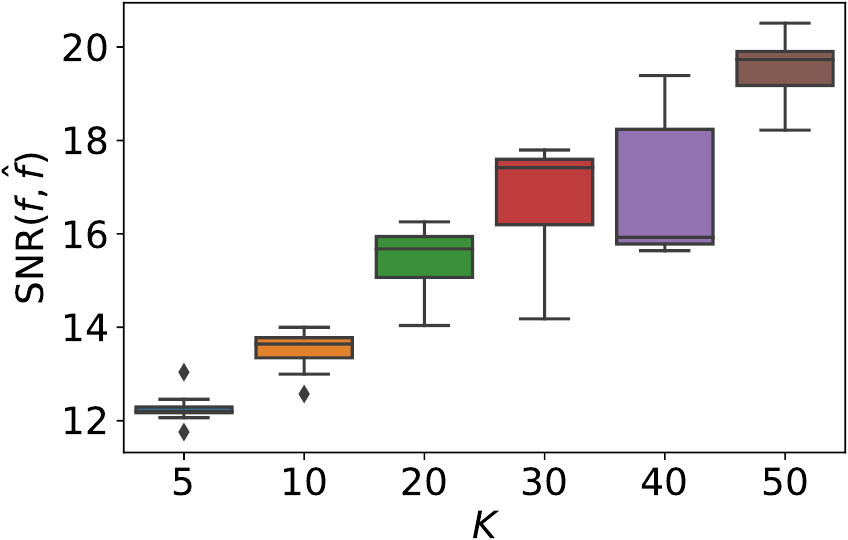}
\includegraphics[width=0.4\textwidth]{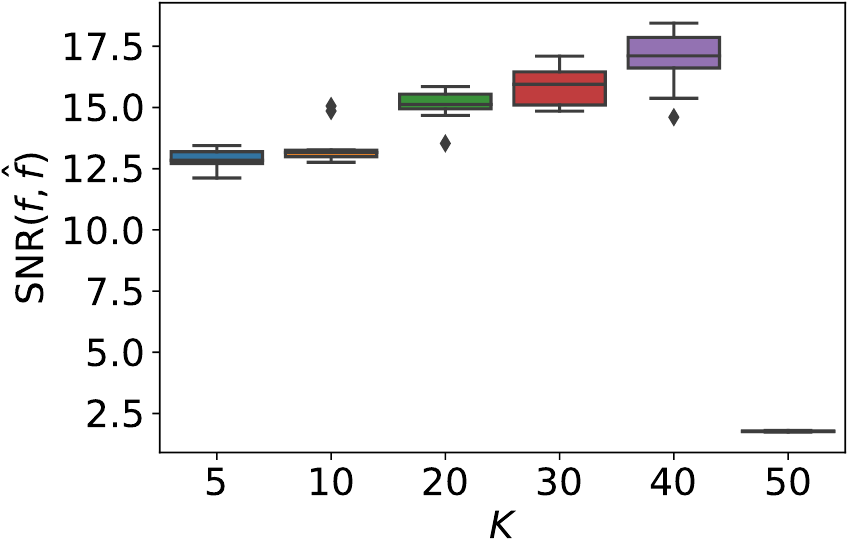}
\end{center}
\caption{SNR performance depending on parameter $K$ for \textit{minnesota} graph, $\sigma=0.01$,  a regular signal $f_{2593-2642}$ (left) and  a non-regular signal $f_{1343-1392}$ (right), over 10 realizations of noise.} 
\label{fig-KselectionSi2}
\end{figure}

\section{Conclusion and future works}

We have introduced a novel technique to efficiently perform graph Fourier analysis. This technique uses functional calculus to perform Fourier analysis on different subsets of the graph Laplacian spectrum. In this paper, we have demonstrated that localization in the spectrum provides interesting improvements in theoretical results for some graph signal analysis tasks. New estimators of the noise level were introduced, taking advantage of the convenient modelling of the denoising problem given by localization, and for which concentration results were proved. Localization allows also to study theoretically the denoising procedure with wavelets, and fits with the design of many well-known techniques (\emph{e.g.} tight frames for graph analysis). Through many experiments, we have validated that localization techniques introduced in this paper improve on state-of-the-art methods for several standard tasks.

Although we provide a rule of thumb to choose a partition $I_\L = \sqcup_{1\leq k\leq K} I_k$ for which denoising results show good performances, experiments suggest that our \textit{elbow rule} is not optimal in most cases. There is certainly an interesting topic in searching for a suitable partition $I_\L = \sqcup_{1\leq k\leq K} I_k$ that would be more adapted to a specific task (\emph{e.g.} denoising). To extend the current work, it would also be interesting to consider other common tasks in GSP, such as de-convolution or in-painting.

\section{Proofs} \label{sec:proofs}

\begin{proof}[Proof of Proposition \ref{prp:block-distrib}]
\begin{enumerate}
\item We have $\tilde f_k = f_k + \xi_k$, where $\xi_k=\sum_{\ell:\lambda_\ell\in I_k} \widehat \xi(\ell) \chi_\ell$. Random variables $(\widehat \xi(\ell))_\ell$ are all distributed as $\mathcal N(0,\sigma^2)$ and independent by orthogonality of the eigenbasis $(\chi_\ell)_\ell$. In particular for $k\neq k^\prime$, vectors $\xi_k$ and $\xi_{k^\prime}$ are independent as expressions involving variables $\widehat \xi(\ell)$ over disjoint subsets $I_k$ and $I_{k^\prime}$. Thus the random variables $(c_k)_{1\leq k\leq K}$ are also independent.
\item When $n_k = n_{k^\prime}$, $\xi_k$ and $\xi_{k^\prime}$ are identically distributed and the result follows from Equality (\ref{equ-noisyNorm}).

When $n_k \neq n_{k^\prime}$, we have $\mathbb E(c_k) \neq \mathbb E(c_{k^\prime})$ as the following equality holds for all $1\leq k\leq K$:
\begin{displaymath}
\mathbb E(c_{k}) = \frac{\| f_k \|_2^2}{n_k} + \sigma^2.
\end{displaymath}
\item Since $(\widehat \xi(\ell))_\ell$ are independent normal variables $\mathcal N(0,\sigma^2)$, the statement is clear from the expression $c_k = \frac{1}{n_k} \sum_{\ell:\lambda_\ell\in I_k} | \widehat \xi(\ell) |^2$.
\end{enumerate}
\end{proof}

The following lemma is useful for the proof of Proposition \ref{concentration-noiseLevel}.
\begin{lem}
\label{lem-binomial}
Let $Z \sim \mathcal B(n,p)$ for some parameters $n \geq 1$ and $p\leq 1/2$. Then
  \begin{displaymath}
    \mathbf P(Z \geq \lceil n/2 \rceil) \leq \exp \left( \frac{n}{2} \ln( 4p(1-p)) \right).
  \end{displaymath}
\end{lem}

\begin{proof}
A simple consequence of \cite[Theorem 1]{Che:52} implies that for all $n \geq 1$ and all $a \geq p$
  \begin{displaymath}
    \mathbf P(Z \geq na) \leq \left [ \frac{(1-p)^{1-a}}{1-a} \left ( \frac{1-a}{a} p \right )^a \right ]^n.
  \end{displaymath}
  Now the result follows since $na=\lceil n/2 \rceil$ implies $\frac12 \leq a \leq \frac12+\frac1n$ so that
  \begin{displaymath}
  \frac{(1-p)^{1-a}}{1-a} \left ( \frac{1-a}{a} p \right )^a \leq \frac{(1-a)^{a-1}}{a^a}\sqrt{p(1-p)} \leq \sqrt{4p(1-p)}.
  \end{displaymath}
\end{proof}

\begin{proof}[Proof of Proposition \ref{concentration-noiseLevel}]
\begin{enumerate}
\item For $k$ such that $I_k \subset \overline{I_f}$, we have $c_k = \frac{\sigma^2}{n_k} \Gamma_{n_k}$, which follows the $\Gamma(\frac{n_k}{2}, \frac{2\sigma^2}{n_k})$ distribution. Then concentration inequalities for $\widehat \sigma_{\rm mean}(c)^2$ are a direct consequence of Theorem 2.57 in~\cite{BerDelRio:15}, applied with $a_k = \frac{n_k}{2}$ and $b_k = \frac{2\sigma^2}{n_k}$.\\

\item For all $k=1, \ldots, K_f$, we define
  \begin{equation*}
    \gamma^-_k = \Phi_{\Gamma_{n_0}}^{-1} \circ \Phi_{\Gamma_{n_k}} \left (\frac{n_k}{\sigma^2} c_k \right ) \quad \textrm{and} \quad \gamma^+_k = \Phi_{\Gamma_{n_\infty}}^{-1} \circ \Phi_{\Gamma_{n_k}} \left (\frac{n_k}{\sigma^2} c_k \right ).
  \end{equation*}
  As a matter of fact, $(\gamma^-)_{k=1, \ldots, K_f}$ and $(\gamma_k^+)_{k=1, \ldots, K_f}$ are two sequences of \emph{i.i.d.} random variables with $\gamma_1^- \sim \chi^2(n_0)$ and $\gamma^+_1 \sim \chi^2(n_\infty)$ such that
  \begin{equation*}
    \forall k=1, \ldots, K_f, \quad \gamma_k^- \leq \frac{n_k}{\sigma^2} c_k \leq \gamma^+_k \quad \textrm{almost surely}.
  \end{equation*}
  Then, for all $t>0$,
  \begin{align}
    \mathbb P \left ( \widehat \sigma^2_{\rm med} \geq \beta^{-1} \sigma^2 + 2\sigma^2 \beta^{-1} t \right ) & = \mathbb P \left ( \sum_{k=1}^{K_f} \indic_{\left \{ c_k \geq \beta^{-1} \sigma^2+2\sigma^2 \beta^{-1} t \right \}} \geq \left \lceil \frac{K_f}{2} \right \rceil \right ) \nonumber
    \\  
    \label{equ-indicatorsTrick}
    & \leq \mathbb P \left ( \sum_{k=1}^{K_f} \indic_{ \left \{ \gamma_k^+ \geq n_\infty + 2 n_\infty t \right \}} \geq \left \lceil \frac{K_f}{2} \right \rceil \right ).
  \end{align}
  Similarly, for all $t \in (0,1)$,
  \begin{align}
  \label{equ-indicatorsTrick2}
    \mathbb P \left ( \widehat \sigma^2_{\rm med} \leq \beta \sigma^2-\sigma^2 \beta t \right ) 
    & \leq \mathbb P \left ( \sum_{k=1}^{K_f} \indic_{ \left \{ \gamma_k^- \leq n_0 - n_0 t \right \}} \geq \left \lceil \frac{K_f}{2} \right \rceil \right ).
  \end{align}

To conclude, apply Lemma~\ref{lem-binomial} to Inequalities~\eqref{equ-indicatorsTrick} and~\eqref{equ-indicatorsTrick2} to obtain our result.
\end{enumerate}
\end{proof}

\begin{proof}[Proof of Proposition \ref{prp-orderedTerms}]
The concentration bound of Equation \eqref{equ-concentrationBound} implies that
\begin{align*}
\mathbb P(b_k \geq b_\ell)& = \mathbb P(n_k b_k \geq n_k b_\ell) = \mathbb P \left(  \sigma^{-2} \left ( \frac{n_k}{n_\ell} \| \xi_\ell \|_2^2 - \| \xi_k \|_2^2 \right ) \leq \sigma^{-2} \left (\|f_k\|_2^2 + 2 \langle f_k, \xi_k \rangle \right )  \right ) \\
& \geq \frac{\alpha}{2}+\mathbb P \left (  \sigma^{-2} \left ( \frac{n_k}{n_\ell} \|\xi_\ell \|_2^2 - \|\xi_k \|_2^2 \right ) \leq \sigma^{-2} \left (\|f_k\|_2^2 + 2 t_{\alpha, \sigma} \|f_k \|_2 \right ) \right).
\end{align*}
Since subsets $I_k$ and $I_\ell$ are disjoint, random variables $\|\xi_k \|_2^2$ and $\|\xi_\ell \|_2^2$ are independent. Thus, $\frac{n_k}{n_\ell} \|\xi_\ell \|_2^2 - \|\xi_k \|_2^2$ is distributed as $\frac{n_k}{n_\ell} \Gamma_{n_\ell} - \Gamma_{n_k}$ where $\Gamma_{n_k}$ and $\Gamma_{n_\ell}$ are independent random variables with $\Gamma_{n_k} \sim \chi^2(n_k)$ and $\Gamma_{n_\ell} \sim \chi^2(n_\ell)$. Therefore, the statement of Proposition \ref{prp-orderedTerms} follows.
\end{proof}

\begin{proof}[Proof of Proposition \ref{prp-denoiseTrivial}]
First, the following equalities hold:
\begin{displaymath}
f-\widetilde{f}= f-\widetilde{f}_I  +  \widetilde{f}_I - \widetilde{f}= (f-\widetilde{f})_I  +  \widetilde{f}_{\overline{I}}.
\end{displaymath}
As $(f-\widetilde{f})_I$ and $\widetilde{f}_{\overline{I}}$ are orthogonal vectors, it follows that
\begin{displaymath}
\left \| f-\widetilde{f} \right \|_2^{2} = \left \| f-\widetilde{f}_I \right \|_2^{2} + \left \| \widetilde{f}_{\overline{I}} \right \|_2^{2}.
\end{displaymath}
It remains to notice that $\mathbb{E}( \n \widetilde{f}_{\overline{I}} \n^{2} ) = \sigma^{2} \vert \overline{I} \cap \mathrm{sp}(\L) \vert $.
\end{proof}

\begin{proof}[Proof of Lemma \ref{lem-pk<alpha}]
By Equation \eqref{equ-noisyNorm} and the concentration bound of Equation \eqref{equ-concentrationBound}, it follows that
\begin{displaymath}
\begin{split}
p_k & = \mathbb P( \sigma^2 \Gamma_{n_k} > \|\xi_{k}\|_2^2 + \|f_k\|_2^2 + 2 \langle f_k, \xi_k\rangle )\\
&
\leq \frac{\alpha}{2} + \mathbb P \left( \sigma^2 \Gamma_{n_k} > \|\xi_{k}\|_2^2 + \|f_k\|_2^2 - 2 \|f_k\|_2 t_{\alpha / 2,\sigma} \right) \\
& = \frac{\alpha}{2}  + \mathbb P \left( \sigma^2 (\Gamma_{n_k} - \Gamma_{n_k}^\prime) >  \|f_k\|_2^2 - 2 \|f_k\|_2 t_{\alpha / 2,\sigma} \right)  \\
& = \frac{\alpha}{2} + 1 - \Phi_{\Gamma_{n_k} - \Gamma_{n_k}^\prime}( \theta(f_k,\alpha,\sigma) ),
\end{split}
\end{displaymath}
where $\theta(f_k,\alpha,\sigma) = \sigma^{-2}( \|f_k\|_2 - 2 t_{\alpha / 2,\sigma})  \|f_k\|_2 $. Consequently, $1 - \Phi_{\Gamma_{n_k} - \Gamma_{n_k}^\prime}( \theta(f_k,\alpha,\sigma) ) \leq \alpha/2$ and $p_k \leq \alpha$.
\end{proof}

\begin{proof}[Proof of Lemma \ref{lem-pk>alpha}]
Using an estimate on the $\chi^2(n_k)$ tail distribution and independence of $\Gamma_{n_k}$ and $\Gamma_{n_k}^\prime = \sigma^{-2} \| \xi_k \|_2^2$, it follows
\begin{displaymath}
\begin{split}
  p_k & \geq \mathbb P \left( \sigma^2 \Gamma_{n_k} >  \|\xi_k\|_2^2 + \|f_k\|_2^2 + 2\langle \xi_k, f_k \rangle, \| \xi_k \|_2^2 \leq t_{\beta, k} \right ) \\
  & \geq \mathbb P \left( \sigma^2 \Gamma_{n_k} >  \|\xi_k\|_2^2 + \|f_k\|_2^2 + 2 \|\xi_k\|_2 \|f_k\|_2, \| \xi_k \|_2^2 \leq t_{\beta, k} \right) \\
  & = \mathbb P \left( \sigma^2 \Gamma_{n_k} >  ( \|f_k\|_2 + \sqrt{t_{\beta, k}} )^2,\| \xi_k \|_2^2 \leq t_{\beta, k} \right) \\
  & = \mathbb P \left( \sigma^2 \Gamma_{n_k} >  ( \|f_k\|_2 + \sqrt{t_{\beta, k}})^2, \sigma^2 \Gamma_{n_k}^\prime \leq t_{\beta, k} \right) \\
  & \geq \mathbb P \left( \sigma^2 \Gamma_{n_k} >  ( \|f_k\|_2 + \sqrt{ t_{\beta, k}  } )^2  \right) (1  -  \beta) \\
&\geq \frac{\alpha}{1 - \beta} \times (1 - \beta)=\alpha.
\end{split}
\end{displaymath}
\end{proof}

\begin{proof}[Proof of Proposition \ref{prp-denois-unknownI}]
\begin{enumerate}
\item To prove Inequality \eqref{equ-FourierSupportError}, first observe that $f = \sum_{k: I_k\subset I_f} f_k$ so that
\begin{displaymath}
f-f_I = \sum_{k: I_k\subset I_f} f_k - \sum_{k: p_k \leq \alpha} f_k.
\end{displaymath}
The summands which are not present in both terms are exactly those satisfying either $I_k\subset I_f$ and $p_k>\alpha$ or $I_k \cap I_f = \emptyset$ and $p_k\leq \alpha$. Noting that $f_k = 0$ when $I_k \cap I_f = \emptyset$, it comes
\begin{displaymath}
\|f-f_I\|_2^2 = \sum_{k: I_k\subset I_f, p_k > \alpha} \|f_k\|_2^2.
\end{displaymath}
Applying Lemma \ref{lem-pk<alpha} for all indices $1\leq k \leq K$ satisfying $p_k > \alpha$, one deduce
\begin{displaymath}
\| f_k \|_2 < t_{\alpha / 2, \sigma}  +  \sqrt{ t_{\alpha / 2, \sigma}^2  +  \left( \sigma \Phi_{\Gamma_{n_k} - \Gamma_{n_k}'}^{-1}( 1- \frac{\alpha}{2} ) \right)^2  }.
\end{displaymath}
from which, since $n_k=n_1$ for all $k$, Inequality~\eqref{equ-FourierSupportError} follows.

\item Since $\sigma^{-2}||\xi_k||_2^2$ is distributed as a $\chi^2(n_1)$ random variable, the second Inequality~\eqref{equ-FourierNoiseError} follows
\begin{displaymath}
\mathbb E \|f_I - \widetilde f_I \|_2^2 = \sum_{k: p_k\leq \alpha} \mathbb E \|\xi_k\|_2^2= |\{ k, p_k \leq \alpha \}| n_1 \sigma^2.
\end{displaymath}
\end{enumerate}
\end{proof}

\begin{proof}[Proof of Theorem \ref{thm-oracleBound}]
Since threshold value is $t_2 = 0$ on $\overline{I}$, $\widehat f = \widehat f_I$. Then, clearly $f_{\overline{I}} (f_I - \widehat f_I) = 0$ almost surely so that
\begin{displaymath}
\mathbb E  \|f-\widehat f\|_2^2= \mathbb E \|f-\widehat f_I\|_2^2= \mathbb E  \|f - f_I + f_I - \widehat f_I \|_2^2= \mathbb E \|f - f_I \|_2^2 + \mathbb \| f_I - \widehat f_I \|_2^2.
\end{displaymath}
Applying Theorem 3 from \cite{gobel2018construction} to $\mathbb  E \| f_I-\widehat{f_I} \|_2^2$ yields our statement.
\end{proof}

\begin{proof}[Proof of Proposition \ref{prp-tightFrames}]
Recalling that, for any function $g$ defined on $\spL$ and any subset $I\subset I_{\L}$,
\begin{displaymath}
\sum_n | \langle \sqrt{g}(\L_{I})\delta_n , f \rangle |^2 = \| \sqrt{g}(\L_{I})f \|_2^2 = \langle g(\L_{I})f, f \rangle
\end{displaymath}
it follows by Equations~\eqref{equ-parsevalFrame} and~\eqref{equ-correspondence}.
\begin{multline*}
  \sum_{n,k} |\langle\varphi_{n,k}, f\rangle|^2  +  |\langle\psi_{1,n,k}, f\rangle|^2  +  |\langle\psi_{2,n,k}, f\rangle|^2 \\
  =  \sum_{k}  \langle\psi_0(\L_{I_k}) f, f\rangle  +  \langle\psi_k(\L_{I_k}) f, f\rangle +  \langle\psi_k(\L_{I_{k+1}}) f, f\rangle \\
  =  \langle\psi_0(\L) f, f\rangle  + \sum_k \langle\psi_k(\L) f, f\rangle =  \| f \|_2^2
\end{multline*}
\end{proof}

\begin{proof}[Proof of Proposition \ref{prp-spectrumAdaptedFrames}]
Remarking that $\widehat g_j(\lambda)=\psi(s_j \omega_0(\lambda) )$ with $s_j = e^{-j \omega_{\gamma,J,R}}$, Equation~\eqref{eq-spectrumAdaptedTightFrame} implies that
\begin{equation}
\label{equ-equ7}
g_{m,j}  =  \sum_l  \psi(s_j \omega_0(\lambda_l))  \widehat \delta_m(l) \chi_l.
\end{equation}
Setting $J_j = [C^{-1} e^{(j-R)\omega_{\gamma,J,R}}  ,  C^{-1} e^{j\omega_{\gamma,J,R}} ]$ and recalling that ${\rm supp}(\widehat g) = [ -R\omega_{\gamma,J,R}, 0]$, it follows that $\lambda \in \mathrm{supp}(\widehat g_j)$ if and only if $\omega_0(\lambda) \in J_j$ if and only if $s_j\omega_0(\lambda) \in J_0$. Moreover, $J_0=\sqcup_k I_k$ and $J_j=s_j^{-1} J_0$ yield $J_j=\sqcup_k s_j^{-1} I_k$ with $s_j^{-1} I_k \cap s^{-1}_{j^\prime} I_{k^\prime}$ excepted when $j=j^\prime$ ad $k=k^\prime$. Consequently, Equation~\eqref{equ-equ7} can be reformulated as
\begin{equation*}
g_{m,j} = \sum_{1\leq k\leq R-1} \sum_{\ell\in I_k}  \psi(s_j \omega_0(\lambda_\ell))  \widehat \delta_m(\ell) \chi_\ell.
\end{equation*}
 
\end{proof}

\bibliography{references}

\begin{thebibliography}{10}

\bibitem{anis2014towards}
Aamir Anis, Akshay Gadde, and Antonio Ortega.
\newblock Towards a sampling theorem for signals on arbitrary graphs.
\newblock In {\em 2014 IEEE International Conference on Acoustics, Speech and
  Signal Processing (ICASSP)}, pages 3864--3868. IEEE, 2014.

\bibitem{behjat2016signal}
Hamid Behjat, Ulrike Richter, Dimitri Van De~Ville, and Leif S{\"o}rnmo.
\newblock Signal-adapted tight frames on graphs.
\newblock {\em IEEE Trans. Signal Process.}, 64(22):6017--6029, 2016.

\bibitem{belkin2008towards}
Mikhail Belkin and Partha Niyogi.
\newblock Towards a theoretical foundation for laplacian-based manifold
  methods.
\newblock {\em J. Comput. Syst. Sci.}, 74(8):1289--1308, 2008.

\bibitem{BerDelRio:15}
Bernard Bercu, Bernard Delyon, and Emmanuel Rio.
\newblock {\em Concentration inequalities for sums and martingales}.
\newblock SpringerBriefs in Mathematics. Springer, Cham, 2015.

\bibitem{BerHar:16}
Tyrus Berry and John Harlim.
\newblock Variable bandwidth diffusion kernels.
\newblock {\em Appl. Comput. Harmon. Anal.}, 40(1):68--96, 2016.

\bibitem{candes2007sparsity}
Emmanuel Candes and Justin Romberg.
\newblock Sparsity and incoherence in compressive sampling.
\newblock {\em Inverse problems}, 23(3):969, 2007.

\bibitem{candes2006stable}
Emmanuel~J Candes, Justin~K Romberg, and Terence Tao.
\newblock Stable signal recovery from incomplete and inaccurate measurements.
\newblock {\em Communications on Pure and Applied Mathematics: A Journal Issued
  by the Courant Institute of Mathematical Sciences}, 59(8):1207--1223, 2006.

\bibitem{chen2015discrete}
Siheng Chen, Rohan Varma, Aliaksei Sandryhaila, and Jelena Kova{\v{c}}evi{\'c}.
\newblock Discrete signal processing on graphs: Sampling theory.
\newblock {\em IEEE transactions on signal processing}, 63(24):6510--6523,
  2015.

\bibitem{chen2008multi}
Xiao-yun Chen and Yan-yan Zhan.
\newblock Multi-scale anomaly detection algorithm based on infrequent pattern
  of time series.
\newblock {\em Journal of Computational and Applied Mathematics},
  214(1):227--237, 2008.

\bibitem{Che:52}
Herman Chernoff.
\newblock A measure of asymptotic efficiency for tests of a hypothesis based on
  the sum of observations.
\newblock {\em Ann. Math. Statistics}, 23:493--507, 1952.

\bibitem{chung1997spectral}
Fan~RK Chung and Fan~Chung Graham.
\newblock {\em Spectral graph theory}.
\newblock Number~92. American Mathematical Soc., 1997.

\bibitem{CoiLaf:06}
Ronald~R. Coifman and St\'{e}phane Lafon.
\newblock Diffusion maps.
\newblock {\em Appl. Comput. Harmon. Anal.}, 21(1):5--30, 2006.

\bibitem{coifman2006diffusion}
Ronald~R Coifman and Mauro Maggioni.
\newblock Diffusion wavelets.
\newblock {\em Appl. Comput. Harmon. Anal.}, 21(1):53--94, 2006.

\bibitem{coulhon2012heat}
Thierry Coulhon, Gerard Kerkyacharian, and Pencho Petrushev.
\newblock Heat kernel generated frames in the setting of dirichlet spaces.
\newblock {\em Journal of Fourier Analysis and Applications}, 18(5):995--1066,
  2012.

\bibitem{crovella2003graph}
Mark Crovella and Eric Kolaczyk.
\newblock Graph wavelets for spatial traffic analysis.
\newblock In {\em IEEE INFOCOM 2003. Twenty-second Annual Joint Conference of
  the IEEE Computer and Communications Societies (IEEE Cat. No. 03CH37428)},
  volume~3, pages 1848--1857. IEEE, 2003.

\bibitem{davis2011university}
Timothy~A. Davis and Yifan Hu.
\newblock The {U}niversity of {F}lorida sparse matrix collection.
\newblock {\em ACM Trans. Math. Software}, 38(1):Art. 1, 25, 2011.

\bibitem{de2019data}
Basile de~Loynes, Fabien Navarro, and Baptiste Olivier.
\newblock Data-driven thresholding in denoising with spectral graph wavelet
  transform.
\newblock {\em J. Comput. Appl. Math.}, 389:113319, 12, 2021.

\bibitem{defferrard2016convolutional}
Micha{\"e}l Defferrard, Xavier Bresson, and Pierre Vandergheynst.
\newblock Convolutional neural networks on graphs with fast localized spectral
  filtering.
\newblock In {\em Advances in neural information processing systems}, pages
  3844--3852, 2016.

\bibitem{di2016efficient}
Edoardo Di~Napoli, Eric Polizzi, and Yousef Saad.
\newblock Efficient estimation of eigenvalue counts in an interval.
\newblock {\em Numer. Linear Algebra Appl.}, 23(4):674--692, 2016.

\bibitem{donoho2006compressed}
David~L Donoho.
\newblock Compressed sensing.
\newblock {\em IEEE Transactions on information theory}, 52(4):1289--1306,
  2006.

\bibitem{eldar2009beyond}
Yonina~C Eldar and Tomer Michaeli.
\newblock Beyond bandlimited sampling.
\newblock {\em IEEE signal processing magazine}, 26(3):48--68, 2009.

\bibitem{fan2019spectrum}
Li~Fan, David~I Shuman, Shashanka Ubaru, and Yousef Saad.
\newblock Spectrum-adapted polynomial approximation for matrix functions.
\newblock In {\em ICASSP 2019-2019 IEEE International Conference on Acoustics,
  Speech and Signal Processing (ICASSP)}, pages 4933--4937. IEEE, 2019.

\bibitem{GarGerHeiSle:20}
Nicol\'{a}s Garc\'{\i}a~Trillos, Moritz Gerlach, Matthias Hein, and Dejan
  Slep\v{c}ev.
\newblock Error estimates for spectral convergence of the graph {L}aplacian on
  random geometric graphs toward the {L}aplace-{B}eltrami operator.
\newblock {\em Found. Comput. Math.}, 20(4):827--887, 2020.

\bibitem{gavish2010multiscale}
Matan Gavish, Boaz Nadler, and Ronald~R Coifman.
\newblock Multiscale wavelets on trees, graphs and high dimensional data:
  theory and applications to semi supervised learning.
\newblock In {\em Proceedings of the 27th International Conference on
  International Conference on Machine Learning}, pages 367--374, 2010.

\bibitem{GinKol:06}
Evarist Gin\'{e} and Vladimir Koltchinskii.
\newblock Empirical graph {L}aplacian approximation of {L}aplace-{B}eltrami
  operators: large sample results.
\newblock In {\em High dimensional probability}, volume~51 of {\em IMS Lecture
  Notes Monogr. Ser.}, pages 238--259. Inst. Math. Statist., Beachwood, OH,
  2006.

\bibitem{girault2018irregularity}
Benjamin Girault, Antonio Ortega, and Shrikanth~S Narayanan.
\newblock Irregularity-aware graph fourier transforms.
\newblock {\em IEEE Trans. Signal Process.}, 66(21):5746--5761, 2018.

\bibitem{gobel2018construction}
Franziska G{\"o}bel, Gilles Blanchard, and Ulrike von Luxburg.
\newblock Construction of tight frames on graphs and application to denoising.
\newblock In {\em Handbook of Big Data Analytics}, pages 503--522. Springer,
  2018.

\bibitem{hammond2011wavelets}
David~K Hammond, Pierre Vandergheynst, and R{\'e}mi Gribonval.
\newblock Wavelets on graphs via spectral graph theory.
\newblock {\em Appl. Comput. Harmon. Anal.}, 30(2):129--150, 2011.

\bibitem{HeiAudLux:05}
Matthias Hein, Jean-Yves Audibert, and Ulrike von Luxburg.
\newblock From graphs to manifolds---weak and strong pointwise consistency of
  graph {L}aplacians.
\newblock In {\em Learning theory}, volume 3559 of {\em Lecture Notes in
  Comput. Sci.}, pages 470--485. Springer, Berlin, 2005.

\bibitem{herley1999minimum}
Cormac Herley and Ping~Wah Wong.
\newblock Minimum rate sampling and reconstruction of signals with arbitrary
  frequency support.
\newblock {\em IEEE Transactions on Information Theory}, 45(5):1555--1564,
  1999.

\bibitem{hutchinson1990stochastic}
M.~F. Hutchinson.
\newblock A stochastic estimator of the trace of the influence matrix for
  {L}aplacian smoothing splines.
\newblock {\em Comm. Statist. Simulation Comput.}, 19(2):433--450, 1990.

\bibitem{Kat:95}
Tosio Kato.
\newblock {\em Perturbation theory for linear operators}.
\newblock Classics in Mathematics. Springer-Verlag, Berlin, 1995.
\newblock Reprint of the 1980 edition.

\bibitem{leonardi2011wavelet}
Nora Leonardi and Dimitri Van De~Ville.
\newblock Wavelet frames on graphs defined by fmri functional connectivity.
\newblock In {\em 2011 IEEE International Symposium on Biomedical Imaging: From
  Nano to Macro}, pages 2136--2139. IEEE, 2011.

\bibitem{leonardi2013tight}
Nora Leonardi and Dimitri Van De~Ville.
\newblock Tight wavelet frames on multislice graphs.
\newblock {\em IEEE Trans. Signal Process.}, 61(13):3357--3367, 2013.

\bibitem{lin2016approximating}
Lin Lin, Yousef Saad, and Chao Yang.
\newblock Approximating spectral densities of large matrices.
\newblock {\em SIAM review}, 58(1):34--65, 2016.

\bibitem{lu2008theory}
Yue~M Lu and Minh~N Do.
\newblock A theory for sampling signals from a union of subspaces.
\newblock {\em IEEE transactions on signal processing}, 56(6):2334--2345, 2008.

\bibitem{mahadevan2008fast}
Sridhar Mahadevan.
\newblock Fast spectral learning using lanczos eigenspace projections.
\newblock In {\em Proc. AAAI Conference on Artificial Intelligence, 2008},
  pages 1472--1475, 2008.

\bibitem{MarCoi:19}
Nicholas~F. Marshall and Ronald~R. Coifman.
\newblock Manifold learning with bi-stochastic kernels.
\newblock {\em IMA J. Appl. Math.}, 84(3):455--482, 2019.

\bibitem{meyer1985principe}
Yves Meyer.
\newblock Principe d’incertitude, bases hilbertiennes et algebres
  d’operateurs.
\newblock {\em S{\'e}minaire Bourbaki}, 662:1985--1986, 1985.

\bibitem{ortega2018graph}
Antonio Ortega, Pascal Frossard, Jelena Kova{\v{c}}evi{\'c}, Jos{\'e}~MF Moura,
  and Pierre Vandergheynst.
\newblock Graph signal processing: Overview, challenges, and applications.
\newblock {\em Proceedings of the IEEE}, 106(5):808--828, 2018.

\bibitem{perraudin2017stationary}
Nathana{\"e}l Perraudin and Pierre Vandergheynst.
\newblock Stationary signal processing on graphs.
\newblock {\em IEEE Trans. Signal Process.}, 65(13):3462--3477, 2017.

\bibitem{puy2018random}
Gilles Puy, Nicolas Tremblay, R{\'e}mi Gribonval, and Pierre Vandergheynst.
\newblock Random sampling of bandlimited signals on graphs.
\newblock {\em Applied and Computational Harmonic Analysis}, 44(2):446--475,
  2018.

\bibitem{ricaud2019fourier}
Benjamin Ricaud, Pierre Borgnat, Nicolas Tremblay, Paulo Gon{\c{c}}alves, and
  Pierre Vandergheynst.
\newblock Fourier could be a data scientist: From graph fourier transform to
  signal processing on graphs.
\newblock {\em Comptes Rendus Physique}, 20(5):474--488, 2019.

\bibitem{Ros:97}
Steven Rosenberg.
\newblock {\em The {L}aplacian on a {R}iemannian manifold}, volume~31 of {\em
  London Mathematical Society Student Texts}.
\newblock Cambridge University Press, Cambridge, 1997.
\newblock An introduction to analysis on manifolds.

\bibitem{sandryhaila2014discrete}
Aliaksei Sandryhaila and Jose~MF Moura.
\newblock Discrete signal processing on graphs: Frequency analysis.
\newblock {\em IEEE Trans. Signal Process.}, 62(12):3042--3054, 2014.

\bibitem{sardellitti2017graph}
Stefania Sardellitti, Sergio Barbarossa, and Paolo Di~Lorenzo.
\newblock On the graph fourier transform for directed graphs.
\newblock {\em IEEE J. Sel. Top. Signal Process.}, 11(6):796--811, 2017.

\bibitem{segarra2015aggregation}
Santiago Segarra, Antonio~G Marques, Geert Leus, and Alejandro Ribeiro.
\newblock Aggregation sampling of graph signals in the presence of noise.
\newblock In {\em 2015 IEEE 6th International Workshop on Computational
  Advances in Multi-Sensor Adaptive Processing (CAMSAP)}, pages 101--104. IEEE,
  2015.

\bibitem{shuman2013emerging}
David~I Shuman, Sunil~K Narang, Pascal Frossard, Antonio Ortega, and Pierre
  Vandergheynst.
\newblock The emerging field of signal processing on graphs: Extending
  high-dimensional data analysis to networks and other irregular domains.
\newblock {\em IEEE Signal Process. Mag.}, 30(3):83--98, 2013.

\bibitem{shuman2016vertex}
David~I Shuman, Benjamin Ricaud, and Pierre Vandergheynst.
\newblock Vertex-frequency analysis on graphs.
\newblock {\em Appl. Comput. Harmon. Anal.}, 40(2):260--291, 2016.

\bibitem{shuman2011chebyshev}
David~I Shuman, Pierre Vandergheynst, and Pascal Frossard.
\newblock Chebyshev polynomial approximation for distributed signal processing.
\newblock In {\em 2011 International Conference on Distributed Computing in
  Sensor Systems and Workshops (DCOSS)}, pages 1--8. IEEE, 2011.

\bibitem{shuman2015spectrum}
David~I. Shuman, Christoph Wiesmeyr, Nicki Holighaus, and Pierre Vandergheynst.
\newblock Spectrum-adapted tight graph wavelet and vertex-frequency frames.
\newblock {\em IEEE Trans. Signal Process.}, 63(16):4223--4235, 2015.

\bibitem{silver1994densities}
RN~Silver and H~R{\"o}der.
\newblock Densities of states of mega-dimensional hamiltonian matrices.
\newblock {\em Int. J. Mod. Phys. C}, 5(04):735--753, 1994.

\bibitem{Sin:06}
A.~Singer.
\newblock From graph to manifold {L}aplacian: the convergence rate.
\newblock {\em Appl. Comput. Harmon. Anal.}, 21(1):128--134, 2006.

\bibitem{susnjara2015accelerated}
Ana Susnjara, Nathanael Perraudin, Daniel Kressner, and Pierre Vandergheynst.
\newblock Accelerated filtering on graphs using lanczos method.
\newblock {\em arXiv preprint arXiv:1509.04537}, 2015.

\bibitem{tanaka2018spectral}
Yuichi Tanaka.
\newblock Spectral domain sampling of graph signals.
\newblock {\em IEEE Transactions on Signal Processing}, 66(14):3752--3767,
  2018.

\bibitem{tanaka2020generalized}
Yuichi Tanaka and Yonina~C Eldar.
\newblock Generalized sampling on graphs with subspace and smoothness priors.
\newblock {\em IEEE Transactions on Signal Processing}, 68:2272--2286, 2020.

\bibitem{tanaka2014m}
Yuichi Tanaka and Akie Sakiyama.
\newblock $ m $-channel oversampled graph filter banks.
\newblock {\em IEEE Trans. Signal Process.}, 62(14):3578--3590, 2014.

\bibitem{tremblay:tel-01078956}
Nicolas Tremblay.
\newblock {\em {Networks and signal : signal processing tools for network
  analysis}}.
\newblock Theses, {Ecole normale sup{\'e}rieure de lyon - ENS LYON}, October
  2014.

\bibitem{LuxBelBou:08}
Ulrike von Luxburg, Mikhail Belkin, and Olivier Bousquet.
\newblock Consistency of spectral clustering.
\newblock {\em Ann. Statist.}, 36(2):555--586, 2008.

\bibitem{vNe:41}
John von Neumann.
\newblock Distribution of the ratio of the mean square successive difference to
  the variance.
\newblock {\em Ann. Math. Statistics}, 12:367--395, 1941.

\bibitem{wang1994calculating}
Lin-Wang Wang.
\newblock Calculating the density of states and optical-absorption spectra of
  large quantum systems by the plane-wave moments method.
\newblock {\em Phys. Rev. B}, 49(15):10154, 1994.

\end{thebibliography}

\end{document}